\documentclass[12pt]{article}

\usepackage[utf8]{inputenc}
\usepackage[T1]{fontenc}
\usepackage{eucal}
\usepackage{microtype}
\usepackage{csquotes}
\usepackage{amsthm}
\usepackage{mathtools}
\usepackage{amssymb,amsfonts,stmaryrd,upgreek}
\usepackage{tikz}
\usetikzlibrary{cd}
\usepackage[citestyle=numeric-comp,sorting=nyt,sortcites, backend=biber, giveninits=true, maxnames=99]{biblatex}
\usepackage{enumitem}
\setlist[enumerate,1]{label={\roman*)}}
\usepackage[unicode=true, pdfusetitle, colorlinks=true, allcolors=blue, bookmarks=false]{hyperref}
\usepackage{authblk} 
\usepackage[all]{xy}   
\usepackage[left=2.5cm,right=2.5cm,top=3cm,bottom=3.5cm]{geometry}
\usepackage[colorinlistoftodos]{todonotes} 

\newbox{\myorcidaffilbox}
\sbox{\myorcidaffilbox}{\large\includegraphics[height=1.7ex]{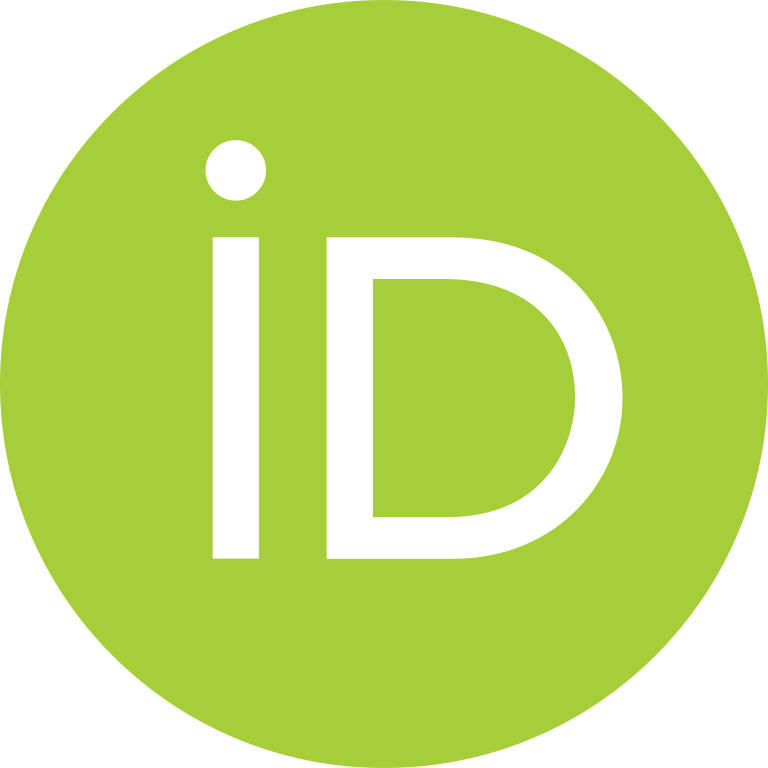}}
\newcommand{\orcid}[1]{%
  \href{https://orcid.org/#1}{\usebox{\myorcidaffilbox}}}

\DeclareUnicodeCharacter{05C3}{**********}

\theoremstyle{plain}
\newtheorem{theorem}{Theorem}
\newtheorem*{theorem*}{Theorem}
\newtheorem{lemma}[theorem]{Lemma}
\newtheorem*{lemma*}{Lemma}
\newtheorem{proposition}[theorem]{Proposition}
\newtheorem*{proposition*}{Proposition}
  
\newtheorem*{corollary*}{Corollary}
\theoremstyle{definition}
\newtheorem{definition}{Definition}
\newtheorem*{definition*}{Definition}
\newtheorem{example}{Example}
\newtheorem*{example*}{Example}
\theoremstyle{remark}
\newtheorem{remark}{Remark}
\newtheorem*{remark*}{Remark}

\newtheorem*{conjecture*}{Conjecture}

\newtheorem*{problem*}{Problem}

\newcommand*{\RR}{\mathbb{R}}

\newcommand*{\dd}{\mathrm{d}}
\newcommand*{\dds}{\dd_{S}}
\newcommand*{\ddst}{\dd_{\tilde S}}

\newcommand*{\contr}[1]{\iota_{#1}}
\newcommand*{\liedv}[1]{\mathcal{L}_{#1}}

\DeclareMathOperator{\Leg}{Leg}

\let\Im\Image

\DeclareMathOperator{\hor}{hor}

\renewcommand{\noeqref}[1]{\textcolor{red}{\Huge Numerar!}}

\bibliography{biblio}

\AtEveryBibitem{
  \ifentrytype{online}{}{
    \ifentrytype{thesis}{}{
        \clearfield{url}
        \clearfield{urldate}
      }
  }
  \ifentrytype{unpublished}{}{ 
  \clearfield{note}
  }
  \clearfield{language}
}

\DeclareSourcemap{
  \maps[datatype=bibtex]{
    \map[overwrite=true]{
      \step[fieldset=urldate, null]
      \step[fieldset=language, null]
      \step[fieldset=address, null]
     \step[fieldset=pagetotal, null]
    }
  }
}

\title{Geometric Hamilton-Jacobi theory for systems with external forces}
\author[1,2]{\orcid{0000-0002-8028-2348} Manuel de León}
\author[1]{\orcid{0000-0002-2368-5853} Manuel Lainz}
\author[1]{\orcid{0000-0002-9620-9647} Asier López-Gordón\thanks{Author to whom correspondence should be addressed: \href{mailto:asier.lopez@icmat.es}{asier.lopez@icmat.es}}
}
\affil[1]{Instituto de Ciencias Matemáticas (CSIC-UAM-UC3M-UCM) \protect\\
Calle Nicolás Cabrera, 13-15, Campus Cantoblanco, UAM, 28049 Madrid, Spain}
\affil[2]{Real Academia de Ciencias Exactas, Físicas y Naturales\protect\\
Calle Valverde, 22, 28004, Madrid, Spain }

\date{\vspace{-5ex}}

\mathtoolsset{showonlyrefs}

\begin{document}

\maketitle

\begin{abstract}
\noindent
In this paper, we develop a Hamilton-Jacobi theory for forced Hamiltonian and Lagrangian systems. We study the complete solutions, particularize for Rayleigh systems and present some examples. Additionally, we present a method for the reduction and reconstruction of the Hamilton-Jacobi problem for forced Hamiltonian systems with symmetry. Furthermore, we consider the reduction of the Hamilton-Jacobi problem for a \v{C}aplygin system to the Hamilton-Jacobi problem for a forced Lagrangian system. 
\end{abstract}






\section{Introduction}
The classical formulation \cite{goldstein_mecanica_1987,abraham_foundations_2008,arnold_mathematical_1978} of the Hamilton-Jacobi problem for a Hamiltonian system on $T^*Q$ consists in looking for a function $S$ on $Q\times \RR$, called the principal function (also known as the generating function), such that
\begin{equation}
  \frac{\partial S} {\partial t} + H \left(q^i, \frac{\partial S} {\partial q^i}  \right)  = 0, \label{eq_H-J_classical_S}
\end{equation}
where $H:T^*Q\to \RR$ is the Hamiltonian function. With the \emph{ansatz} $S(q^i, t) = W(q^i) - tE$, where $E$ is a constant, the equation above reduces to
\begin{equation}
  H \left(q^i, \frac{\partial W} {\partial q^i}  \right) = E, 
  \label{eq_H-J_classical_W}
\end{equation}
where $W:Q\to \RR$ is the so-called characteristic function. Both Eqs.~\eqref{eq_H-J_classical_S} and \eqref{eq_H-J_classical_W} are known as the Hamilton-Jacobi equation.

Despite the difficulties to solve a partial differential equation instead of a system of ordinary differential equations, i.e., to solve the Hamilton-Jacobi equation instead of Hamilton equations, the Hamilton-Jacobi theory provides a remarkably powerful method to integrate the dynamics of many Hamiltonian systems. In particular, for a completely integrable system, if one knows as much constants of the motion in involution as degrees of freedom of the system, one can obtain a complete solution of the Hamilton-Jacobi problem and completely solve the Hamiltonian system, or, in other words, reduce it to quadratures \cite{grillo_non-commutative_2021,grillo_extended_2021,grillo_hamiltonjacobi_2016}.

Geometrically, the Hamilton-Jacobi equation \eqref{eq_H-J_classical_W} can be written as
\begin{equation}
  \dd (H \circ W) = 0,
\end{equation}
where $\dd W$ is a 1-form on $Q$. This 1-form transforms the integral curves of a vector field $X_H^{\dd W}$ on $Q$ into integral curves of the dynamical vector field $X_H$ on $TQ$ (the latter satisfying Hamilton equations). This geometric procedure \cite{carinena_geometric_2006,abraham_foundations_2008} has been extended to many other different contexts, such as nonholonomic systems \cite{carinena_geometric_2006,carinena_geometric_2010,iglesias-ponte_towards_2008,ohsawa_nonholonomic_2011,leon_linear_2010}, singular Lagrangian systems \cite{de_leon_hamilton-jacobi_2013,de_leon_hamiltonjacobi_2012,leok_hamiltonjacobi_2012}, higher-order systems \cite{colombo_geometric_2014}, field theories \cite{de_leon_geometric_2008,de_leon_hamiltonjacobi_2014,vitagliano_geometric_2012,de_leon_hamiltonjacobi_2020,de_leon_hamilton-jacobi_2017,campos_hamilton-jacobi_2015,zatloukal_classical_2016} or contact systems \cite{de_leon_hamiltonjacobi_2021}. An unifying Hamilton-Jacobi theory for almost-Poisson manifolds was developed in reference \cite{de_leon_hamilton-jacobi_2014}. The Hamilton-Jacobi theory has also been generalized to Hamiltonian systems with non-canonical symplectic structures \cite{martinez-merino_hamiltonjacobi_2006}, non-Hamiltonian systems \cite{rashkovskiy_hamilton-jacobi_2020} locally conformally symplectic manifols \cite{esen_hamiltonjacobi_2021}, Nambu-Poisson and Nambu-Jacobi manifolds \cite{de_leon_geometric_2017,de_leon_geometric_2019}, Lie algebroids \cite{leok_dirac_2012} and implicit differential systems \cite{esen_hamiltonjacobi_2018}.
The applications of Hamilton-Jacobi theory include the relation between classical and quantum mechanics \cite{carinena_hamilton-jacobi_2009,marmo_hamilton-jacobi_2009,budiyono_quantization_2012}, information geometry \cite{ciaglia_hamilton-jacobi_2017,ciaglia_hamilton-jacobi_2017-1}, control theory \cite{sakamoto_analysis_2002} and the study of phase transitions \cite{kraaij_hamilton-jacobi_2021}.

In the same fashion, in this paper we develop a Hamilton-Jacobi theory for systems with external forces. This paper is the natural continuation of our previous paper about symmetries and constants of the motion of systems with external forces \cite{de_leon_symmetries_2021} (see also \cite{lopez-gordon_geometry_2021}).
Mechanical systems with external forces (so-called forced systems) appear commonly in engineering and describe certain physical systems with dissipation \cite{de_leon_symmetries_2021,lopez-gordon_geometry_2021,esen_geometrical_2021}. Moreover, they emerge after a process of reduction in a nonholonomic \v{C}aplygin system \cite{cantrijn_geometry_2002,garcia-naranjo_geometry_2020,ohsawa_nonholonomic_2011,bloch_nonholonomic_1996,koiller_reduction_1992,iglesias-ponte_towards_2008}. A particular type of external forces are the so-called Rayleigh forces \cite{de_leon_symmetries_2021,lopez-gordon_geometry_2021,bloch_euler-poincare_1996}, i.e., forces that an be written as the derivative of a ``potential'' with respect to the velocities. Forced systems on a Lie group have been studied in reference \cite{bloch_euler-poincare_1996}.

This paper is organized as follows. In Section \ref{section_preliminaries} we recall the geometric concepts we will make use of. In Section \ref{section_HJ_Hamiltonian} we develop a Hamilton-Jacobi theory for Hamiltonian systems with external forces. We consider the complete solutions, particularized for Rayleigh forces and discuss some examples. The analogous theory for Lagrangian systems with external forces is described in Section \ref{section_HJ_Lagrangian}. In Section \ref{section_reduction} we present a method for the reduction and reconstruction of solutions of the Hamilton-Jacobi problem for forced Hamiltonian systems with symmetry. Finally, Section \ref{section_Chaplygin} is devoted to the reduction of \v{C}aplygin systems to forced Lagrangian systems, in order to obtain solutions of the forced Hamilton-Jacobi problem and reconstruct solutions of the nonholonomic Hamilton-Jacobi problem.

\section{Preliminaries} \label{section_preliminaries}
This paper is a natural continuation of our previous paper \cite{de_leon_symmetries_2021}.
Let us briefly recall the notations and results we will make use of.
See also \cite{lopez-gordon_geometry_2021}.

\subsection{Semibasic forms and fibred morphims}\label{section_morphisms}

    Consider a fibre bundle $\pi: E\to M$. Let us recall \cite{leon_methods_1989,abraham_foundations_2008,godbillon_geometrie_1969} that a 1-form $\beta$ on $E$ is called \emph{semibasic} if 
    \begin{equation}
        \beta(Z)=0
    \end{equation}
    for all vertical vector fields $Z$ on $E$. If $(x^i,y^a)$ are fibred (bundle) coordinates, then the vertical vector fields are locally generated by $\{\partial/\partial y^a\}$. So $\beta$ is a semibasic 1-form if it is locally written as
    \begin{equation}
        \beta=\beta_i(x,y) \dd x^i.
    \end{equation}

We shall particularize this definition for the cases of tangent and cotangent bundles. In what follows, let $Q$ be an $n$-dimensional differentiable manifold. Given a morphism of fibre bundles
\begin{center}
\begin{tikzcd}
 TQ \arrow[rr, "D"] \arrow[rd, "\tau_q"'] &   & T^*Q \arrow[ld, "\pi_Q"] \\
                                                          & Q &                         
\end{tikzcd},
\end{center}
one can define an associated semibasic 1-form \cite{godbillon_geometrie_1969,leon_methods_1989} $\beta_D$ on $TQ$ by 
\begin{equation}
    \beta_D(v_q)(u_{v_q})=\left\langle D(v_q),T\tau_Q(u_{v_q})\right\rangle,
\end{equation}
where $v_q\in T_qQ,\ u_{v_q}\in T_{v_q}(TQ)$.

If locally $D$ is given by
\begin{equation}
    D(q^i,\dot q^i)=(q^i,D_i(q,\dot q)),
\end{equation}
then
\begin{equation}
    \beta_D=D_i(q,\dot q)\dd q^i.
\end{equation}

Conversely, given a semibasic 1-form $\beta$ on $TQ$, one can define the following morphism of fibre bundles:
\begin{equation}
\begin{aligned}
    &D_\beta: TQ\to T^*Q,\\
    &\left\langle D_\beta(v_q),w_q\right\rangle=\beta(v_q)(u_{w_q}),
\end{aligned}
\end{equation}
for every $v_q,w_q\in T_qQ,\ u_{w_q}\in T_{w_q}(TQ)$, with $T\tau_Q(u_{w_q})=w_q$. In local coordinates, if
\begin{equation}
\beta=\beta_i(q,\dot q)\dd q^i,
\end{equation}
then
\begin{equation}
    D_\beta(q^i,\dot q^i)=\left(q^i,\beta_i(q^i,\dot{q}^i)\right).
\end{equation}
Here $(q^i, \dot{q}^i)$ are bundle coordinates in $TQ$.

So there exists a one-to-one correspondence between semibasic 1-forms on $TQ$ and fibred morphisms from $TQ$ to $T^*Q$.

\subsection{Hamiltonian mechanics}
An external force is geometrically interpreted as a semibasic 1-form on $T^*Q$. A Hamiltonian system with external forces (so called \emph{forced Hamiltonian system}) $(H,\beta)$ is given by a Hamiltonian function $H:T^*Q\to \mathbb{R}$ and a semibasic 1-form $\beta$ on $T^*Q$. Let $\omega_Q=-\dd \theta_Q$ be the canonical symplectic form of $T^*Q$.
Locally these objects can be written as
\begin{equation}
\begin{aligned}
&\theta_Q=p_i\dd q^i,\\
&\omega_Q=\dd q^i\wedge \dd p_i,\\
&\beta=\beta_i(q,p) \dd q^i,\\
&H=H(q,p),
\label{local_symplectic}
\end{aligned}
\end{equation}
where $(q^i,p_i)$ are bundle coordinates in $T^*Q$.

The dynamics of the system is given by the vector field $X_{H,\beta}$, defined by
\begin{equation}
\contr{X_{H,\beta}}\omega_Q=\dd H+\beta. \label{Hamiltonian_dynamics_eq}
\end{equation}
If $X_H$ is the Hamiltonian vector field for $H$, that is,
\begin{equation}
\contr{X_H}\omega_Q=\dd H, \label{Hamiltonian_VF}
\end{equation}
and $Z_\beta$ is the vector field defined by
\begin{equation}
\contr{Z_\beta}\omega_Q=\beta,
\end{equation}
then we have
\begin{equation}
X_{H,\beta}=X_H+Z_\beta.
\end{equation}
Locally, the above equations can be written as 
\begin{equation}
\begin{aligned}
&X_H=\frac{\partial  H} {\partial p_i }\frac{\partial  } {\partial  q^i}
-\frac{\partial H } {\partial  q^i} \frac{\partial  } {\partial  p_i}
 \label{Hamiltonian_VF_local},
\\
&\beta=\beta_i \dd q^i,\\
&Z_\beta=-\beta_i \frac{\partial  } {\partial p_i} ,\\
&X_{H,\beta}=\frac{\partial  H} {\partial p_i }\frac{\partial  } {\partial  q^i}-\left(\frac{\partial H } {\partial  q^i}+\beta_i\right)
\frac{\partial  } {\partial  p_i}.
\end{aligned}
\end{equation}
Then, a curve $(q^i(t), p_i (t)$ in $T^*Q$ is an integral curve of $X_{H,\beta}$ if and only if it satisfies the \emph{forced Hamilton equations}
\begin{equation}
\begin{aligned}
&\frac{\mathrm{d}q^i} {\mathrm{d}t}=\frac{\partial  H} {\partial p_i },\\
&\frac{\mathrm{d}p_i} {\mathrm{d}t}=-\left(\frac{\partial  H} {\partial q^i }+\beta_i\right).
\label{forced_Hamilton_eqs}
\end{aligned}
\end{equation}
Let us recall that the \emph{Poisson bracket} is the bilinear operation
\begin{equation}
\begin{aligned}
\left\{\cdot,\cdot  \right\}: C^\infty(M,\RR)\times C^\infty(M,\RR) &\to C^\infty(M,\RR)\\
\left\{f,g  \right\} &=\omega(X_f,X_g),
\end{aligned} \label{Poisson_bracket}
\end{equation}
with $X_f, X_g$ the Hamiltonian vector fields associated to $f$ and $g$, respectively.

\begin{definition}
Let $(H, \beta)$ be a forced Hamiltonian system on $T^*Q$. A function $f$ on $T^*Q$ is called a \emph{constant of the motion} (or a \emph{conserved quantity}) if
\begin{equation}
  X_{H, \beta}(f) = 0, \label{conserved_quantity_Hamiltonian}
\end{equation}
or, equivalently, $f$ is constant along the solutions of the forced Hamilton equations \eqref{forced_Hamilton_eqs}.
\end{definition}

\subsection{Lagrangian systems with external forces}
The \emph{Poincaré-Cartan 1-form} on $TQ$ associated with the Lagrangian function $L:TQ\to \RR$ is
\begin{equation}
  \theta_L=S^*(\dd L), \label{Poincare_Cartan_1}
\end{equation}
where $S^*$ is the adjoint operator of the vertical endomorphism on $TQ$, which is locally
\begin{equation}
  S = \dd q^i \otimes \frac{\partial  } {\partial  \dot{q}^i}. \label{vertical_endomorphism}
\end{equation}
The \emph{Poincaré-Cartan 2-form} is $\omega_L=-\dd \theta_L$, so locally
\begin{equation}
  \omega_L = \dd q^i \wedge \dd \left( \frac{\partial L} {\partial  \dot{q}^i}  \right). \label{Poincare_Cartan_2}
\end{equation}
One can easily verify that $\omega_L$ is symplectic if and only if $L$ is regular, that is, if the Hessian matrix
\begin{equation}
  \left( W_{ij}  \right) = \left( \frac{\partial^2 L  } {\partial \dot{q}^i \partial \dot{q}^j }  \right)
\end{equation}
is invertible. 
The \emph{energy} of the system is given by
\begin{equation}
  E_L=\Delta(L)-L,
\end{equation}
where $\Delta$ is the Liouville vector field,
\begin{equation}
  \Delta= \dot{q}^i \frac{\partial  } {\partial  \dot{q}^i}.
\end{equation}
Similarly to the Hamiltonian framework, an external force is represented by a semibasic 1-form $\alpha$ on $TQ$. Locally,
\begin{equation}
  \alpha = \alpha_i (q, \dot{q}) \dd q^i.
\end{equation}
The dynamics of the forced Lagrangian system $(L, \alpha)$ is given by
\begin{equation}
  \contr{\xi_{L,\alpha}} \omega_L = \dd E_L + \alpha, \label{dynamics_Lagrangian}
\end{equation}
that is, the integral curves of the \emph{forced Euler-Lagrange vector field} $\xi_{L,\alpha}$ satisfy the forced Euler-Lagrange equations:
\begin{equation}
\frac{\mathrm{d}}{\mathrm{d} t}\left(\frac{\partial L}{\partial \dot{q}^i}\right) 
  - \frac{\partial L}{\partial q^i}
  = -\alpha_i .
\label{forced_continuous_EL}
\end{equation}
We can write $\xi_{L, \alpha}=\xi_L + \xi_\alpha$, where
\begin{equation}
\begin{aligned} 
  &\contr{\xi_{L}} \omega_L = \dd E_L,\\
  &\contr{\xi_{\alpha}} \omega_L =  \alpha.
\end{aligned}
\end{equation}
This vector field is a \emph{second order differential equation} (\emph{SODE}), that is, 
\begin{equation}
  S(\xi_{L,\alpha}) = \Delta.
\end{equation}

\begin{definition}
Let $(L,\alpha)$ be a forced Lagrangian system on $TQ$. A function $f$ on $TQ$ is called a \emph{constant of the motion} (or a \emph{conserved quantity}) if
\begin{equation}
  \xi_{L, \alpha}(f) = 0, \label{conserved_quantity_Lagrangian}
\end{equation}
or, equivalently, $f$ is constant along the solutions of the forced Euler-Lagrange equations \eqref{forced_continuous_EL}.
\end{definition}

\subsection{Rayleigh forces}

An external force $\bar{R}$ is said to be \emph{of Rayleigh type} (or simply \emph{Rayleigh} for short) \cite{de_leon_symmetries_2021,lopez-gordon_geometry_2021} if there exists a function $\mathcal{R}$ on $TQ$ such that
\begin{equation}
  \bar{R}=S^*(\dd \mathcal{R}),
\end{equation}
which can be locally written as
\begin{equation}
  \bar{R} 
  = \bar R_i \dd q^i
  = \frac{\partial \mathcal{R}} {\partial \dot{q}^i} \dd q^i.
\end{equation}
This function $\mathcal{R}$ is called the \emph{Rayleigh dissipation function} (or the \emph{Rayleigh potential}). In other words, the fibred morphism $D_{\bar R}: TQ\to T^*Q$ associated with $\bar{R}$ is given by the fibre derivative of $\mathcal{R}$ (see reference \cite{bloch_euler-poincare_1996}), namely
\begin{equation}
   D_{\bar R} = \mathbb{F} \mathcal{R}: \left(q^i, \dot q^i  \right) 
   \mapsto \left(q^i, \frac{\partial \mathcal R} {\partial \dot q^i}  \right).
 \end{equation} 
A \emph{Rayleigh system} $(L,\mathcal{R})$ is a forced Lagrangian system with Lagrangian function $L$, and with external force $\bar{R}$ generated by the Rayleigh potential $\mathcal{R}$. For a Rayleigh system $(L, \mathcal R)$ with Rayleigh force $\bar R$, the forced Euler-Lagrange vector field is denoted by $\xi_{L, \bar{R}}$, given by
\begin{equation}
  \contr{\xi_{L,\bar R}} \omega_L = \dd E_L + \bar R
  = \dd E_L + S^*\left( \dd \mathcal R  \right).
  \label{dynamics_Rayleigh}
\end{equation}
This vector field can be written as $\xi_{L, \bar R}=\xi_L + \xi_{\bar R}$, where
\begin{equation}
  \contr{\xi_{\bar R}} \omega_L =  \bar R,
  \label{vf_Rayleigh}
\end{equation}
and the forced Euler-Lagrange equations \eqref{forced_continuous_EL} can be written as
\begin{equation}
  \frac{\mathrm{d}}{\mathrm{d} t}\left(\frac{\partial L}{\partial \dot{q}^i}\right) 
  - \frac{\partial L}{\partial q^i}
  = - \bar{R}_i
  = - \frac{\partial \mathcal R} {\partial \dot q^i}.
  \label{Euler_Lagrange_Rayleigh}
\end{equation}

\begin{remark}
If a Rayleigh potential $\mathcal{R}$ on $TQ$ defines a Rayleigh force $\bar{R}$ on $TQ$, $\mathcal{R}+f$ also defines $\bar{R}$, for any function $f$ on $Q$. In other words, given a Rayleigh force $\bar R$, its associated Rayleigh potential $\mathcal{R}$ is defined up to the addition of a basic function on $TQ$.
\end{remark}

The \emph{vertical differentiation} \cite{leon_methods_1989} $\dds$ on $T^*Q$ is given by
\begin{equation}
  \dds
  =[\contr{S}, d]
  = \contr{S} d - d \contr{S},
\end{equation}
where $\contr{S}$ denotes the \emph{vertical derivation}, given by
\begin{equation}
\begin{aligned}
    &\contr{S} f = 0,\\
    &(\contr{S} \omega) (X_1, \ldots, X_p) = \sum_{i=1}^p \omega(X_1, \ldots, S X_i, \ldots, X_p) ,
\end{aligned}
\end{equation}
for any function $f$, any $p$-form $\omega$ and any vector fields $X_1, \ldots, X_p$ on $TQ$. 
In particular,
\begin{equation}
  \dds f = S^* (\dd f),
\end{equation}
for any function $f$ on $TQ$. We can then write a Rayleigh force as
\begin{equation}
  \bar R = \dds {\mathcal {R}}.
\end{equation}

 A \emph{linear Rayleigh force} $\bar{R}$ is a Rayleigh force for which $\mathcal{R}$ is a quadratic form in the velocities, namely
\begin{equation}
  \mathcal{R} (q, \dot{q}) = \frac{1}{2} R_{ij} (q) \dot{q}^i \dot{q}^j,
\end{equation}
where $R_{ij}$ is symmetric and non-degenerate, and hence the Rayleigh force is
\begin{equation}
  \bar{R} = R_{ij}(q) \dot{q}^i \dd q^j.
\end{equation}
 In such a case, one can define an associated \emph{Rayleigh tensor} $R\in T^*Q \times T^*Q$, given by
\begin{equation}
  R = R_{ij} \dd q^i \otimes \dd q^j.
\end{equation}
 A \emph{linear Rayleigh system} $(L, R)$ is a Rayleigh system such that $\bar{R}$ is a linear Rayleigh force with Rayleigh tensor $R$.

The \emph{Legendre transformation} is a mapping $\Leg: TQ\to T^*Q$ such that the diagram
\begin{center}
\begin{tikzcd}
TQ \arrow[rr, "\Leg"] \arrow[rd, "\tau_q"'] &   & T^*Q \arrow[ld, "\pi_Q"] \\
                                                          & Q &                      
\end{tikzcd}
\end{center}
commutes. Here $\tau_q$ and $\pi_Q$ are the canonical projections on $Q$.
Locally,
\begin{equation}
  \Leg: (q^i,\dot{q}^i) \mapsto (q^i,p_i),
\end{equation}
with $p_i=\partial L/\partial \dot{q}^i$. In what follows, let us assume that the Lagrangian $L$ is \emph{hyperregular}, i.e., that $\Leg$ is a (global) diffeomorphism.

\subsection{Dissipative bracket}

\begin{definition}
    The \emph{dissipative bracket} is a bilinear map $[\cdot,\cdot]:C^\infty(TQ)\times C^\infty(TQ)\to C^\infty(TQ)$ given by
    \begin{equation}
      [f,g] = (SX_f)(g), \label{dissipative_bracket}
    \end{equation}
    where $S$ is the vertical endomorphism, and $X_f$ is the Hamiltonian vector field associated to $f$ on $(TQ,\omega_L)$, namely,
    \begin{equation}
      \contr{X_f} \omega_L = \dd f.
    \end{equation}
\end{definition}

\begin{proposition}
The dissipative bracket $[\cdot,\cdot]$ on $(TQ,\omega_L)$ verifies the following properties:
\begin{enumerate}
\item $[f,g]=[g,f]$ (it is symmetric),
\item $[f,gh]=[f,h]g+[f,g]h$ (``Leibniz rule''),
\end{enumerate}
for all functions $f,g$ on $TQ$.
\end{proposition}
\begin{proof}
In local coordinates,
\begin{equation}
  [f,g] = W^{ij} \frac{\partial f} {\partial \dot{q}^j} \frac{\partial g} {\partial \dot{q}^i},
\end{equation}
as one can derive from Eqs.~\eqref{vertical_endomorphism} and \eqref{Poincare_Cartan_2}.
Here $(W^{ij})$ is the inverse matrix of the Hessian matrix $(W_{ij})$ of the Lagrangian $L$. From this expression both assertions can be easily proven.
\end{proof}

Since the dissipative bracket is bilinear and verifies the Leibniz rule, it is a derivation, or a so-called Leibniz bracket \cite{ortega_dynamics_2004}.

\begin{proposition}
  Consider a Rayleigh system $(L,\mathcal{R})$ on $(TQ,\omega_L)$.
  A function $f$ on $TQ$ is a constant of the motion of $(L,\mathcal{R})$ if and only if
  \begin{equation}
    \left\{f,E_L  \right\} - [f,\mathcal{R}] = 0,
    \label{double_bracket_Rayleigh}
  \end{equation}
  where $\left\{\cdot,\cdot  \right\}$ is the Poisson bracket \eqref{Poisson_bracket} defined by $\omega_L$.
\end{proposition}
\begin{proof}
As a matter of fact,
\begin{equation}
\begin{aligned} 
    [f,\mathcal{R}]
    &= (S X_f) (\mathcal{R})
    = \contr{SX_f} \dd \mathcal{R}
    = \contr{X_f} (S^* \dd \mathcal{R} )
    = \contr{X_f} \bar{R}\\
    &= \contr{X_f} \contr{\xi_{\bar{R}}} \omega_L
    = -\contr{\xi_{\bar{R}}} \contr{X_f} \omega_L
    = -\contr{\xi_{\bar{R}}} \dd f
    = -\xi_{\bar{R}}(f),
\end{aligned}
\end{equation}
and
\begin{equation}
   \left\{f,E_L  \right\}
   = \omega_L(X_f,\xi_L)
   = \contr{\xi_L}\contr{X_f}\omega_L
   = \xi_L(f),
\end{equation}
so
\begin{equation}\begin{aligned}
   \left\{f,E_L  \right\} - [f,\mathcal{R}]
   = \xi_L(f) + \xi_{\bar{R}}(f)
   = \xi_{L,\bar{R}} (f). \label{eq_proof_bracket_conserved}
\end{aligned}\end{equation}
Here $\xi_{L,\bar{R}}$ is the dynamical vector field given by Eq.~\eqref{dynamics_Rayleigh} and $\xi_{\bar R}$ is given by Eq.~\eqref{vf_Rayleigh}.
In particular by Eq.~\eqref{conserved_quantity_Lagrangian}, the right-hand side of Eq.~\eqref{eq_proof_bracket_conserved} vanishes if and only if $f$ is a constant of the motion.
\end{proof}

\begin{remark}
Other types of dissipative systems, particularly thermodynamical systems, exhibit a ``double bracket'' dissipation, i.e., their dynamics are described in terms of two brackets (in our case the Poisson bracket \eqref{Poisson_bracket} and the dissipative bracket \eqref{dissipative_bracket}). As a matter of fact, the dissipative bracket we defined above has certain similarities with other types of brackets.

The dissipative bracket $[\cdot,\cdot]$ defined above resembles the dissipative bracket $(\cdot,\cdot)$ appearing in the metriplectic framework \cite{coquinot_general_2020,birtea_asymptotic_2007}. Both brackets are symmetric and bilinear. However, the latter requires the additional assumption that $(E_L,f)$ vanishes identically for every function $f$ on $TQ$. Clearly, this requirement does not hold for our dissipative bracket.

On the other hand, our dissipative bracket $[\cdot,\cdot]$ can also be related with the so-called Ginzburg-Landau bracket $[\cdot,\cdot]_{\mathrm{GL}}$ \cite{grmela_dynamics_1997}. This bracket, together with symmetry and bilinearity, satisfies the positivity condition, i.e., $[f,f]_{\mathrm{GL}}\geq 0$ holds pointwisely for all $f$ on $TQ$. As a matter of fact, this holds for our bracket in the case of many relevant Lagrangians. For instance, consider a Lagrangian of the form 
\begin{equation}
  L = \sum_{i} m_i \left(\dot{q}^i\right)^2 - V(q),
\end{equation}
with positive masses $m_i$.
Then
\begin{equation}
  [f,f] = \sum_i m_i \left( \frac{\partial f} {\partial \dot q^i}  \right)^2\geq 0.
\end{equation}
See also reference \cite{bloch_euler-poincare_1996} for various types of systems with double bracket dissipation.
A further research on our dissipative bracket and its applications on thermodynamics will be done elsewhere.




\end{remark}

\subsection{Natural Lagrangians and Hamiltonian Rayleigh forces}

Consider a \emph{natural Lagrangian} $L$ on $TQ$, i.e., a Lagrangian function of the form
\begin{equation}
    L = \frac{1}{2} g_{ij}(q) \dot q^i \dot q^j - V(q), \label{quadratic_Lagrangian}
\end{equation}
where 
\begin{equation}
    g = g_{ij}(q) \dd q^i \otimes \dd q^j
\end{equation}
 is a (pseudo)Riemannian metric on $Q$.
  Clearly, $L$ is regular if and only if $g$ is non-degenerate. As it is well-known, these are the usual Lagrangians in classical mechanics. The Legendre transformation is now linear:
 \begin{equation}
 \begin{aligned}
     \Leg: \left(q^i, \dot{q}^i  \right)&\mapsto \left(q^i, g_{ij} \dot{q}^j  \right),\\
     \Leg^{-1}: \left(q^i, p_i  \right)&\mapsto \left(q^i, g^{ij} p_j  \right),  
 \end{aligned}
 \end{equation}
 where $(g^{ij})=(g_{ij})^{-1}$. In other words, the Legendre transformation consists simply in the ``raising and lowering of indices'' defined by the metric $g$.

Consider a linear Rayleigh system $(L, \mathcal{R})$, where $L$ is natural.
The associated \emph{Hamiltonian Rayleigh potential} $\tilde{\mathcal{R}}$ on $T^*Q$ is given by
\begin{equation}
    \tilde{\mathcal{R}} = \mathcal{R} \circ \Leg^{-1} = \mathcal{R} \circ g^{-1}.
\end{equation}
Similarly, the \emph{Hamiltonian Rayleigh force} $\tilde{R}$ on $T^*Q$ is given by
\begin{equation}
    \tilde{R} 
    = \left(\Leg^{-1}  \right)^* \bar R
    = \left(g^{-1}  \right)^* \bar R.
\end{equation}
When the Lagrangian is regular, the Legendre transformation is well-defined, so we can define a tensor field $\tilde S \in T(T^*Q)\otimes T^*Q$ given by
\begin{equation}
  \tilde{S} =\Leg_* S 
\end{equation}
In particular, if the Lagrangian is natural, then we have
\begin{equation}
  \tilde{S} = g_* S = g_{ij} \frac{\partial  } {\partial p_i} \otimes \dd q^j.
\end{equation}
Hence, the Hamiltonian Rayleigh force $\tilde{R}$ can be expressed in terms of the the Rayleigh potential $\tilde{\mathcal{R}}$ as
\begin{equation}
  \tilde R = \tilde{S}^*(\dd \tilde {\mathcal{R}}).
\end{equation}
We shall omit the adjective Hamiltonian and refer to $\tilde{\mathcal{R}}$ and $\tilde{R}$ as the Rayleigh potential and the Rayleigh force, respectively, if there is no danger of confusion.

Let us introduce the \emph{vertical differentiation} \cite{leon_methods_1989} $\dd _{\tilde S}$ on $T^*Q$ as
\begin{equation}
  \ddst = \contr{\tilde S}\dd -\dd \contr{\tilde S},
\end{equation}
where $\contr{\tilde S}$ is defined analogously to $\contr{S}$ by replacing $S$ with $\tilde S$.
In particular,
\begin{equation}
  \ddst f = \tilde S^*(\dd f),
\end{equation}
for any function $f$ on $T^*Q$. We can then write a Hamiltonian Rayleigh force as
\begin{equation}
  \tilde R = \ddst \tilde {\mathcal {R}}.
\end{equation}

Consider a linear Rayleigh system $(L, R)$, where $L$ is natural.
The associated Hamiltonian Rayleigh potential $\tilde{\mathcal{R}}$ on $T^*Q$ is given by
\begin{equation}
    \tilde{\mathcal{R}} 
                        = \frac{1}{2} R^{ij}(q) p_i p_j,
\end{equation}
where
\begin{equation}
    R^{ij}(q) = g^{ik} g^{jl} R_{kl} (q).
\end{equation}
Similarly, the Hamiltonian Rayleigh force $\tilde{R}$ on $T^*Q$ is given by
\begin{equation}
    \tilde{R} = \left(g^{-1}  \right)_* \bar R = R^i_j  p_i \dd q^j,
\end{equation}
where $R^i_j  = R_{kj} g^{ik}$. 
The associated Hamiltonian Rayleigh tensor $\hat{R}\in TQ \otimes T^*Q$ is given by
\begin{equation}
  \hat{R} = \left(g^{-1}  \right)_* R = R^i_j  \frac{\partial  } {\partial q^i} \otimes \dd q^j.
\end{equation}
The linear Hamiltonian Rayleigh force $\tilde{R}$ can thus be written as
\begin{equation}
  \tilde{R} = \hat{R}^* (\theta_Q).
\end{equation}
This motivates the next definition of linear Rayleigh forces in the Hamiltonian framework, without the need of considering natural Lagrangians, as follows.

\begin{definition}
An external force is called a \emph{linear Hamiltonian Rayleigh force} if it can be written as
\begin{equation}
  \tilde{R} = \hat{R}(\theta_Q)
\end{equation}
for some tensor $\hat{R}\in TQ \otimes T^*Q$. This tensor is called the \emph{Hamiltonian Rayleigh tensor}. A \emph{linear Hamiltonian Rayleigh system} $(H, \hat{R})$ is a forced Hamiltonian system whose external force is a linear Hamiltonian Rayleigh force. When there is no ambiguity, the adjective Hamiltonian will be omitted.
\end{definition}

\begin{remark}
Since $T^*Q$, unlike $TQ$, has not a canonical vertical endomorphism, there is not a natural way to define Hamiltonian non-linear Rayleigh forces, besides Legendre-transforming Lagrangian Rayleigh forces.
\end{remark}

\section{Hamilton-Jacobi theory for systems with external forces} \label{section_HJ_Hamiltonian}
Let $(H,\beta)$ be a forced Hamiltonian system on $T^*Q$. Its dynamical vector field $X_{H,\beta}$ is given by Eq.~\eqref{Hamiltonian_dynamics_eq}.
Given a 1-form $\gamma$ on $Q$ (i.e., a section of $\pi_Q:T^*Q\to Q$), it is possible to project $X_{H,\beta}$ along $\gamma(Q)$, obtaining the vector field 
\begin{equation}
  X_{H,\beta}^\gamma = T\pi_Q \circ X_{H,\beta} \circ \gamma \label{projected_vf}
\end{equation}

on $Q$, so that the following diagram commutes:

\begin{center}
\begin{tikzcd}
T^*Q \arrow[rrr, "{X_{H,\beta}}"] \arrow[dd, "\pi_Q"]                             &  &  & TT^*Q \arrow[dd, "T\pi_Q"] \\
                                                                                  &  &  &                            \\
Q \arrow[rrr, "{X_{H,\beta}^\gamma}"] \arrow[uu, "\gamma", bend left, shift left] &  &  & TQ                        
\end{tikzcd}

\end{center}

\begin{lemma}
The vector fields $X_{H,\beta}$ and $X_{H,\beta}^\gamma$ are $\gamma$-related if and only if $X_{H,\beta}$ is tangent to $\gamma(Q)$.
\end{lemma}

\begin{proof}
By definition, $X_{H,\beta}$ and $X_{H,\beta}^\gamma$ are $\gamma$-related if
\begin{equation}
  T\gamma (X_{H,\beta}^\gamma) = X_{H,\beta} \circ \gamma. \label{def_gamma_related}
\end{equation}
\end{proof}

Therefore, the integral curves of $X_{H,\beta}^\gamma$ are mapped to integral curves of $X_{H,\beta}$ (which satisfy the forced Hamilton equations \eqref{forced_Hamilton_eqs}) via $\gamma$. Indeed, if $\sigma$ is an integral curve of $X_{H,\beta}^\gamma$, then
\begin{equation}
  X_{H,\beta} \circ \gamma \circ \sigma = T\gamma \circ X_{H,\beta}^\gamma \circ \sigma
  = T\gamma \circ \dot{\sigma} = \dot{\overline{\sigma \circ \gamma}},
\end{equation}
so $\gamma\circ \sigma$ is an integral curve of $X_{H,\beta}$. Conversely, if  $\gamma \circ \sigma$ is an integral curve of $X_{H,\beta}$ for every integral curve $\sigma$ of $Y = T\pi_Q \circ X_{H,\beta} \circ \gamma$, then 
\begin{equation}
  T\gamma \circ Y \circ \sigma = T\gamma \circ \dot \sigma 
   = \dot{\overline{\sigma \circ \gamma}} = X_{H,\beta} \circ \gamma \circ \sigma,
\end{equation}
for every integral curve $\sigma$, and hence $X_{H,\beta}$ and $Y$ are $\gamma$-related.

From Eq.~\eqref{Hamiltonian_VF_local}, locally we have that
\begin{equation}
  T\gamma (X_{H,\beta}^\gamma)
   = \frac{\partial H} {\partial p_i} \frac{\partial  } {\partial q^i}
   +\frac{\partial H} {\partial p_i} \frac{\partial \gamma_j} {\partial q^i} \frac{\partial  } {\partial p_j}.
\end{equation}
Then Eq.~\eqref{def_gamma_related} yields
\begin{equation}
  \frac{\partial H} {\partial q^i} + \frac{\partial H} {\partial p_j} \frac{\partial \gamma_i} {\partial q^j} = -\beta_i \circ \gamma.
\end{equation}
In other words,
\begin{equation}
   \frac{\partial H} {\partial q^i} + \frac{\partial H} {\partial p_j} \frac{\partial \gamma_j} {\partial q^i}
   + \beta_i \circ \gamma 
   + \frac{\partial H} {\partial p_j} \left(\frac{\partial \gamma_i} {\partial q^j} 
   - \frac{\partial \gamma_j} {\partial q^i}  \right) = 0 ,
\end{equation}
that is
\begin{equation}
\begin{aligned}
  \dd (H\circ \gamma)  + \gamma^*\beta + \contr{X_{H,\beta}^\gamma} \dd \gamma = 0. 
\end{aligned}
\end{equation}
If $\gamma$ is closed, 
we have
\begin{equation}
  \frac{\partial H} {\partial q^i} + \frac{\partial H} {\partial p_j} \frac{\partial \gamma_j} {\partial q^i} = -\beta_i \circ \gamma,
\end{equation}
that is,
\begin{equation}
  \dd (H\circ \gamma) = -\gamma^*\beta. \label{HJ_eq}
\end{equation}
Let us recall that a \emph{Lagrangian submanifold} $\mathfrak{L}\subset T^*Q$ is a maximal isotropic submanifold, i.e., a submanifold such that $\omega_Q|_{S}=0$ and $\dim S = 1/2 \dim T^*Q = \dim Q$. Clearly, a 1-form $\gamma$ on $Q$ is closed if and only if $\Im \gamma$ is a Lagrangian submanifold. 

\begin{definition}
A 1-form $\gamma$ on $Q$ is called a \emph{solution of the Hamilton-Jacobi problem} for $(H,\beta)$ if:
\begin{enumerate}
\item it is closed,
\item it satisfies equation \eqref{HJ_eq}.
\end{enumerate}
This equation is known as the \emph{Hamilton-Jacobi equation}.
\end{definition}

The results above can be summarized in the following theorem.

\begin{theorem}\label{theorem_HJ}
Let $\gamma$ be a closed 1-form on $Q$. Then the following conditions are equivalent:
\begin{enumerate}
\item $\gamma$ is a solution of the Hamilton-Jacobi problem for $(H,\beta)$,
\item for every curve $\sigma: \RR \to Q$ such that
\begin{equation}
  \dot{\sigma}(t)=T \pi_{Q}\circ X_{H, \beta}\circ \gamma\circ \sigma(t),
   \label{eq_1_theorem_H-J}
\end{equation}
for all $t$, then $\gamma \circ \sigma$ is an integral curve of $X_{H, \beta}$;
\item $\Im \gamma$ is a Lagrangian submanifold of $T^*Q$ and $X_{H,\beta}$ is tangent to it.
\end{enumerate}
\end{theorem}

\begin{remark}
By the Hamilton-Jacobi equation \eqref{HJ_eq}, $\dd H + \beta$ vanishes on $\Im \gamma$, and hence
\begin{equation}
  \dd \beta|_{\Im \gamma} = 0.
\end{equation}
If $\dd \beta$ is non-degenerate, it is a symplectic form and $\Im \gamma$ is a Lagrangian submanifold on $(T^*Q, \dd \beta)$. In general, it is not easy to see whether $\dd \beta$ is non-degenerate or not. However, there is a type of external forces for which this is simple: the Rayleigh forces linear in the momenta. 

\begin{lemma}
Let $\hat R$ be a Hamiltonian Rayleigh tensor on $T^*Q$, and let $\tilde R$ be the associated Hamiltonian Rayleigh force on $T^*Q$. If $\hat{R}$ is non-degenerate, then $\dd \tilde{R}$ is a symplectic form on $T^*Q$.
\end{lemma}

\begin{proof}
We have that
\begin{equation}
  \dd \tilde{R} = \frac{\partial R^k_{\ j}} {\partial q^i} p_k \dd q^i \wedge \dd q^j
                - R^j_{\ i} \dd q^i \wedge \dd p_j,
\end{equation}
so
\begin{equation}
  \contr{X} \dd \tilde R = \left[ \left( \frac{\partial R^k_{\ j}} {\partial q^i  } -\frac{\partial R^k_{\ i}} {\partial q^j  }  \right) p_k X^i - R^i_j  Y_i  \right] \dd q^j
  - R^{j}_{\ i} X^i \dd p_j,
\end{equation}
for a vector field $X = X^i \partial/\partial q^i + Y_i \partial/\partial p_i$ on $T^*Q$. Then $X\in \ker \Omega_{Q,\gamma}$ if and only if
\begin{equation}
\begin{aligned}
  & \left( \frac{\partial R^k_{\ j}} {\partial q^i  } -\frac{\partial R^k_{\ i}} {\partial q^j  }  \right) p_k X^i - R^i_j  Y_i, 
  & R^j_{\ i} X^i = 0. 
\end{aligned}
\end{equation}
If $\hat{R}$ is non-degenerate (in particular, if $R$ is positive-definite), then $X\in \ker \dd \tilde{R}$ if and only if $X=0$, so $\dd \tilde{R}$ is symplectic.
\end{proof}

When $R^i_j $ does not depend on $(q^i)$, we can make the change of bundle coordinates
\begin{equation}
  p_j \mapsto \tilde{p}_j = -R^{i}_{\ j} p_i,
\end{equation}
so that $\dd \tilde{R} = \omega_Q$, and $(q^i, \tilde{p}_i)$ are Darboux coordinates. 

\begin{proposition}
Consider a linear Rayleigh system $(H,\hat{R})$. Suppose that $\hat{R}$ is non-degenerate. Then, a closed 1-form $\gamma$ on $Q$ is a solution of the Hamilton-Jacobi problem for $(H,\tilde{R})$ if and only if $\Im \gamma$ is a Lagrangian submanifold of $(T^*Q, \dd \tilde R)$.
\end{proposition}

For a Rayleigh system $(H, \tilde{\mathcal{R}} )$, the Hamilton-Jacobi equation can also be written as
\begin{equation}
  \gamma^* \left(\dd H + \ddst \tilde {\mathcal{R}}  \right) = 0.
  \label{HJ_Rayleigh}
\end{equation}
We will also refer to the Hamilton-Jacobi problem for $(H, \bar{R})$ as the Hamilton-Jacobi problem for $(H, \mathcal{R})$.
In the case of a linear Rayleigh system $(H, \hat{R})$, we have
\begin{equation}
  \gamma^* \tilde{R} = R^i_j  \gamma_j \dd q^j = \hat{R}^* (\gamma), 
\end{equation}
so the Hamilton-Jacobi equation \eqref{HJ_Rayleigh} can be written as
\begin{equation}
  \dd (H\circ \gamma) + \hat{R}^* (\gamma) = 0. 
  \label{HJ_Rayleigh_linear}
\end{equation}
We will also refer to the Hamilton-Jacobi problem for $(H, \bar{R})$ as the Hamilton-Jacobi problem for $(H, \hat{R})$.
\end{remark}

One can consider a more general problem by relaxing the hypothesis of $\gamma$ being closed.

\begin{definition}
A \emph{weak solution of the Hamilton-Jacobi problem} for $(H,\beta)$ is a 1-form $\gamma$ on $Q$ such that $X_{H,\beta}$ and $X_{H,\beta}^\gamma$ are $\gamma$-related. Here $ X_{H,\beta}^\gamma$ is the vector field defined by \eqref{projected_vf}.
\end{definition}

\begin{proposition}
Consider a 1-form $\gamma$ on $Q$. Then the following statements are equivalent:

\begin{enumerate}
\item $\gamma$ is a weak solution of the Hamilton-Jacobi problem for $(H,\beta)$,
\item $\gamma$ satisfies the equation
\begin{equation}
     \contr{ X_{H,\beta}^\gamma} \dd \gamma = -\dd (H\circ \gamma) - \gamma^* \beta,
     \label{eq_theorem_generalized_H-J}
\end{equation}
\item $X_{H,\beta}$ is tangent to the submanifold $\Im \gamma\subset T^*Q$,
\item if $\sigma:\RR\to Q$ satisfies 
\begin{equation}
  \dot{\sigma}(t)=T \pi_{Q}\circ X_{H, \beta}\circ \gamma\circ \sigma(t), \label{weak_H-J_eq}
\end{equation}
then $\gamma\circ \sigma$ is an integral curve of $X_{H,\beta}$.
\end{enumerate}
\end{proposition}

Let
\begin{equation}
  \eta = \liedv{X_{H,\beta}^\gamma} \gamma + \gamma^*\beta.
\end{equation}
Observe that Eq.~\eqref{weak_H-J_eq} holds if and only if
\begin{equation}
  \eta + \dd \left(H\circ \gamma - \gamma({X_{H,\beta}^\gamma})  \right) = 0.
\end{equation}

\begin{remark}[Local expressions]
Let $\gamma$ be a 1-form $\gamma$ on $Q$. Let $H$ be a Hamiltonian function on $T^*Q$, let $\beta, \tilde{\mathcal R}$ and $\hat R$ be an external force, a Rayleigh potential and a Rayleigh tensor on $T^*Q$, respectively. If $\gamma$ is closed, then
\begin{enumerate}
\item $\gamma$ is a solution of the Hamilton-Jacobi problem for $(H, \beta)$ if and only if
\begin{equation}
  \frac{\partial H} {\partial q^i} + \frac{\partial H} {\partial p_j} \frac{\partial \gamma_j} {\partial q^i} + \beta_i \circ \gamma = 0,\qquad i=1,\ldots, n,
  \label{HJ_local}
\end{equation}
\item $\gamma$ is a solution of the Hamilton-Jacobi problem for $(H, \tilde{\mathcal R})$ if and only if
\begin{equation}
  \frac{\partial H} {\partial q^i} + \frac{\partial H} {\partial p_j} \frac{\partial \gamma_j} {\partial q^i} + \frac{\partial \tilde{\mathcal R}} {\partial p_i} \circ \gamma = 0,\qquad i=1,\ldots, n,
  \label{HJ_local_Rayleigh}
\end{equation}
\item $\gamma$ is a solution of the Hamilton-Jacobi problem for $(H, \hat R)$ if and only if
\begin{equation}
  \frac{\partial H} {\partial q^i} + \frac{\partial H} {\partial p_j} \frac{\partial \gamma_j} {\partial q^i} + R^i_j \gamma_j = 0,\qquad i=1,\ldots, n,
  \label{HJ_local_Rayleigh_linear}
\end{equation}

\end{enumerate}
Moreover, if $\gamma$ is not necessarily closed, it is a weak solution of the Hamilton-Jacobi problem for $(H, \beta)$ if and only if
\begin{equation}
   \frac{\partial H} {\partial q^i} + \frac{\partial H} {\partial p_j} \frac{\partial \gamma_j} {\partial q^i}
   + \beta_i \circ \gamma 
   + \frac{\partial H} {\partial p_j} \left(\frac{\partial \gamma_i} {\partial q^j} 
   - \frac{\partial \gamma_j} {\partial q^i}  \right) = 0.
   \label{HJ_local_weak}
\end{equation}
Similar expressions can be easily found for Rayleigh forces or linear Rayleigh forces.
\end{remark}

\subsection{Complete solutions}
The main interest in the standard Hamilton-Jacobi theory lies in finding a complete family of solutions to the problem \cite{carinena_geometric_2006,de_leon_hamilton-jacobi_2014}. As it is explained below, knowing a complete solution of the Hamilton-Jacobi problem for a forced Hamiltonian system is tantamount to completely integrating the system, namely, there is a constant of the motion for each degree of freedom of the system, and these constants of the motion are in mutual involution.

Consider a forced Hamiltonian system $(H,\beta)$ on $T^*Q$, and assume that $\dim Q=n$.

\begin{definition}\label{def_complete_solutions}
Let $U\subseteq \RR^n$ be an open set. A map $\Phi:Q\times U\to T^*Q$ is called a \emph{complete solution of the Hamilton-Jacobi problem} for $(H,\beta)$ if
\begin{enumerate}
\item $\Phi$ is a local diffeomorphism,
\item for any $\lambda=(\lambda_1,\ldots,\lambda_n)\in U$, the map
\begin{equation}
\begin{aligned}
  \Phi_\lambda: Q&\to T^*Q\\
  q&\mapsto \Phi_\lambda(q)=\Phi(q,\lambda_1,\ldots,\lambda_n)
\end{aligned}
\end{equation}
is a solution of the Hamilton-Jacobi problem for $(H,\beta)$.
\end{enumerate}
\end{definition}

For the sake of simplicity, we shall assume $\Phi$ to be a global diffeomorphism. Consider the functions given by
\begin{equation}
  f_a=\pi_a\circ \Phi^{-1}:T^*Q\to \RR,
\end{equation}
where $\pi_a$ denotes the projection over the $a$-th component of $\RR^n$.

\begin{proposition}
  The functions $f_a$ are constants of the motion. Moreover, they are in involution, i.e., $\left\{f_a,f_b  \right\}=0$, where $\left\{\cdot,\cdot  \right\}$ is the Poisson bracket defined by $\omega_Q$.
\end{proposition}

\begin{proof}
Given $p\in T^*Q$, suppose that $f_a(p)=\lambda_a$ for each $a=1,\ldots,n$. Observe that 
\begin{equation}
  \Im(\Phi_\lambda)= \bigcap_{a=1}^n f_a^{-1}(\lambda_a),
\end{equation}
or, in other words,
\begin{equation}
  \Im(\Phi_\lambda) = \left\{ x\in T^*Q \mid f_a(x)=\lambda_a, a=1,\ldots,n  \right\}.
\end{equation}
By Theorem \ref{theorem_HJ}, $X_{H,\beta}$ is tangent to $\Im \Phi_\lambda$, and hence 
  \begin{equation}
    X_{H,\beta}(f_a)=0
  \end{equation}
for every $a=1,\ldots,n$, that is, $f_1,\ldots,f_n$ are constants of the motion.
In addition,
\begin{equation}
  \left\{f_a,f_b  \right\}|_{\Im (\Phi_\lambda)}
  =\omega_Q(X_{f_a}, X_{f_b})|_{\Im (\Phi_\lambda)} = 0.
\end{equation}
\end{proof}

\begin{example}\label{example_Hamiltonian_drag}
Consider a forced Hamiltonian system $(H,\beta)$, with
\begin{equation}
  H = \frac{1}{2} \sum_{i=1}^n p_i^2,\qquad
  \beta = \sum_{i=1}^n \kappa_i p_i^2 \dd q_i.
\end{equation}
Consider the functions
\begin{equation}
  f_a =  e^{\kappa_aq^a} p_a,\qquad a=1,\ldots,n.
\end{equation}
The dynamics of the system is given by
\begin{equation}
  X_{H,\beta} = p_i \frac{\partial  } {\partial q^i} - \kappa_i p_i^2 \frac{\partial  } {\partial p_i},
\end{equation}
so that
\begin{equation}
  X_{H,\beta} (f_a) = 0,\qquad a=1,\ldots, n,
\end{equation}
and, thus, the functions are constants of the motion. 
Their Hamiltonian vector fields are given by
\begin{equation}
  X_{f_a} = e^{\kappa_a q^a} \left( \frac{\partial  } {\partial q^a} -\kappa p_a \frac{\partial  } {\partial p_a}  \right).
\end{equation}
Clearly,
\begin{equation}
  \left\{f_a,f_b  \right\} = \dd f_a(X_{f_b}) = 0 
\end{equation}
for every $a,b=1,\ldots,n$, so the constants of the motion are in involution. Consider the 1-form $\gamma$ on $Q$ given by
\begin{equation}
  \gamma = \sum_{i=1}^n \lambda_i e^{-\kappa_i q^i} \dd q^i.
\end{equation}
Clearly, $\gamma$ is closed. In fact, it is exact:
\begin{equation}
  \gamma = \dd g,\qquad 
  g = - \sum_{i=1}^n \frac{\lambda_i}{\kappa_i}  e^{-\kappa_i q^i}.
\end{equation}
Moreover,
\begin{equation}
  H\circ \gamma = \frac{1}{2} \sum_{i=1}^n \lambda_i^2 e^{-2\kappa_i q^i},
\end{equation}
and
\begin{equation}
  \gamma^*\beta = \sum_{i=1}^n \kappa_i \lambda_i^2 e^{-2\kappa_i q^i}
                = - \dd (H\circ \gamma).
\end{equation}
Hence $\gamma$ is a complete solution of the Hamilton-Jacobi problem.
When $\kappa_i=0$ for every $i=1,\ldots,n$, 
\begin{equation}
     \gamma  = \sum_{i=1}^n \lambda_i \dd q^i,
\end{equation} 
which is a complete solution for the (conservative) Hamilton-Jacobi problem for $H$.
See Example \ref{Example_H-J_Lagrangian_fluid_n_dimensions} for the Lagrangian counterpart of this example.
\end{example}


\begin{example}[Free particle with a homogeneous linear Rayleigh force] \label{example_Hamiltonian_homogeneous_Rayleigh}
Consider a linear Rayleigh system $(H,\hat{R})$, with
\begin{equation}
  H = \frac{1}{2} g^{ij} p_i p_j,   
\end{equation} 
where $g$ does not depend on $q$ (for instance, $g_{ij}=m_i \delta_{ij}$, with $m_i$ the mass of the $i$-th particle), and 
\begin{equation}
  \hat{R} = R^i_{j} \frac{\partial  } {\partial q^i} \otimes \dd q^j
\end{equation}
where $R^i_j$ does not depend on $q$. Then the Hamilton-Jacobi equation \eqref{HJ_Rayleigh_linear} can be locally written as
\begin{equation}
  \gamma_k \left(g^{jk} \frac{\partial \gamma_j} {\partial q^i} + R_i^k  \right) = 0,
\end{equation}
so a complete solution is given by
\begin{equation}
  \gamma_\lambda = \left(\lambda_i - R_{ij} q^j  \right) \dd q^i,
\end{equation}
where $R_{ij} = g_{jk} R_i^k$. Clearly $\gamma_\lambda$ is closed, in fact, it is exact:
\begin{equation}
  \gamma_\lambda = \dd S_\lambda,\qquad
  S_\lambda = \lambda_i q^i - R_{ij} q^i q^j.
\end{equation}

\end{example}

\section{Hamilton-Jacobi theory for Lagrangian systems with external forces} \label{section_HJ_Lagrangian}
As it has been seen in the previous sections, the natural framework for the Hamilton-Jacobi theory is the Hamiltonian formalism on the cotangent bundle. 
Following Cariñena, Gràcia, Marmo, Martínez, Muñoz-Lecanda and Román-Roy \cite{marmo_hamilton-jacobi_2009,carinena_hamilton-jacobi_2009,carinena_geometric_2006,carinena_geometric_2010,carinena_structural_2016}, we introduce an analogous problem in the Lagrangian formalism on the tangent bundle as follows.

\begin{definition}
A vector field $X$ on $Q$ is called a \emph{solution of the Lagrangian Hamilton-Jacobi problem} for $(L,\alpha)$ if:
\begin{enumerate}
\item $\Leg \circ X$ is a closed 1-form,
\item $X$ satisfies the equation 
\begin{equation}
    \dd(E_L\circ X)=-X^* \alpha. \label{Lagrangian_H-J_equation}
\end{equation}
\end{enumerate}
This equation is known as the \emph{Lagrangian Hamilton-Jacobi equation}. When there is no risk of ambiguity, we shall refer to the Lagrangian Hamilton-Jacobi problem (resp. equation) as simply the Hamilton-Jacobi problem (resp. equation).
\end{definition}

If $\gamma=\Leg \circ X$ is a closed 1-form, then $\Im \gamma$ is a Lagrangian submanifold of $(T^*Q,\omega_Q)$. Therefore $\Im X$ is a Lagrangian submanifold of $(TQ,\omega_L)$. In other words, $\Leg \circ X$ is closed if and only if $X^*\omega_L=0$. Moreover, it is easy to see that $X$ and $\xi_{L,\alpha}$ are $X$-related, that is,
\begin{equation}
  \xi_{L,\alpha}\circ X = TX \circ X. \label{equation_proposition_X_related}
\end{equation} 
Analogously to the Hamiltonian case, one can show the following result (see also Refs.~\cite{iglesias-ponte_towards_2008,carinena_geometric_2006}).

\begin{proposition}
Let $X$ be a vector field on $Q$ that satisfies $X^*\omega_L=0$. Then the following assertions are equivalent:
\begin{enumerate}
\item X is a solution of the Hamilton-Jacobi problem for $(L,\alpha)$,
\item $\Im X$ is a Lagrangian submanifold of $TQ$ invariant by $\xi_{L,\alpha}$,
\item for every curve $\sigma:\RR\to Q$ such that $\sigma$ is an integral curve of $X$, then  $X \circ \sigma: \RR\to TQ$ is an integral curve of $\xi_{L,\alpha}$.
\end{enumerate}
\end{proposition}

\begin{remark}
For a Rayleigh system $(L, \mathcal{R})$, the Hamilton-Jacobi equation \eqref{Lagrangian_H-J_equation} can be written as
\begin{equation}
  X^*(\dd E_L + \dds \mathcal{R}) = 0.
\end{equation}
If $\sigma$ is an integral curve of $X$, then $X\circ \sigma$ is an integral curve of $\xi_{L, \bar R}$, which satisfies the forced Euler-Lagrange equations \eqref{Euler_Lagrange_Rayleigh}.
\end{remark}

As in the Hamiltonian case, one can consider a more general problem by relaxing the hypothesis of $\Leg\circ X$ being closed.
\begin{definition}
A \emph{weak solution of the Hamilton-Jacobi problem} for $(L,\alpha)$ is a vector field $X$ on $Q$ such that $X$ and $\xi_{L,\alpha}$ are $X$-related.
\end{definition}

  Clearly, a weak solution of the Hamilton-Jacobi problem for $(L,\alpha)$ is a solution of the Hamilton-Jacobi problem for $(L,\alpha)$ if and only if $X^*\omega_L=0$.



 

\begin{proposition} Let $X$ be a vector field on $Q$. Then the following statements are equivalent:
\begin{enumerate}
\item X is a solution of the generalized Hamilton-Jacobi problem for $(L,\alpha)$,
\item $X$ satisfies the equation
\begin{equation}
  \contr{X}(X^*\omega_L)= \dd (E_L\circ X) + X^*\alpha, \label{equation_theorem_generalized_solution}
\end{equation}
\item the submanifold $\Im X\subset TQ$ is invariant by $\xi_{L,\alpha}$,
\item if $\sigma:\RR\to Q$ is an integral curve of $X$, then $X\circ \sigma$ is an integral curve of $\xi_{L,\alpha}$.
\end{enumerate}
\end{proposition}

\begin{proof}
The last two assertions are trivial. Let us now prove the equivalence between the first and the second statements.
  From the dynamical equation \eqref{dynamics_Lagrangian} we have
  \begin{equation}
    X^*(\contr{\xi_{L,\alpha}}\omega_L)=X^*(\dd E_L+\alpha) = \dd (E_L\circ X) +X^*\alpha.
  \end{equation}
  Since $X$ and $\xi_{L,\alpha}$ are $X$-related, we can write
  \begin{equation}
    X^*(\contr{\xi_{L,\alpha}}\omega_L)=\contr{X}(X^*\omega_L),
  \end{equation}
  which yields Eq.~\eqref{equation_theorem_generalized_solution}. 

  The proof of the converse is completely analogous to the one of Theorem 1 in Ref.~\cite{carinena_geometric_2006}.
\end{proof}

\subsection{Equivalence between Lagrangian and Hamiltonian Hamilton-Jacobi problems}
Given a forced Lagrangian system $(L, \alpha)$ on $TQ$ (with $L$ hyperregular), one can obtain an associated forced Hamiltonian system $(H,\beta)$ on $T^*Q$, where
\begin{equation}\begin{aligned}
    &H \circ \Leg = E_L,\\
    & \Leg^*\beta = \alpha.
\end{aligned}\end{equation}
Moreover, the dynamical vector fields $\xi_{L,\alpha}$ and $X_{H,\beta}$ (given by Eqs.~\eqref{dynamics_Lagrangian} and \eqref{Hamiltonian_dynamics_eq}, respectively) are $\Leg$-related, i.e.,
\begin{equation}
     T\Leg \circ\ \xi_{L,\alpha} = X_{H,\beta} \circ \Leg.
\end{equation} 

\begin{theorem}
Consider a hyperregular forced Lagrangian system $(L,\alpha)$ on $TQ$, with associated forced Hamiltonian system $(H,\beta)$ on $T^*Q$. Then $X$ is a (weak) solution of the Hamilton-Jacobi problem for $(L,\alpha)$ if and only if $\gamma=\Leg \circ X$ is a (weak) solution of the Hamilton-Jacobi problem for $(H,\beta)$.
\end{theorem}

\begin{proof}
Let $X$ be a weak solution of the Hamilton-Jacobi problem for $(L,\alpha)$.
Then
\begin{equation}
    T\gamma \circ X 
    = T\Leg \circ\ TX \circ X
    = T\Leg \circ\ \xi_{L,\alpha} \circ X
    = X_{H,\beta} \circ \Leg \circ X
    = X_{H,\beta} \circ \gamma,
\end{equation}
since $X$ and $\xi_{L,\alpha}$ are $X$-related.
Composing the left and right hand sides with $T\pi_Q$ from the left, we obtain
\begin{equation}
    X = T\pi_Q \circ X_{H,\beta} \circ \gamma.
\end{equation}
Then $X=X_{H,\beta}^\gamma$, and $\gamma$ is a weak solution of the Hamilton-Jacobi problem for $(H,\beta)$.

Conversely, if $\gamma$ is a solution of the Hamilton-Jacobi problem for $(H,\beta)$, $X$ is $\gamma$-related to $X_{H,\beta}$. Moreover,
\begin{equation}
    \xi_{L,\alpha} \circ \Leg^{-1} = T \left(\Leg^{-1}  \right) \circ X_{H,\beta},
\end{equation}
and hence
\begin{equation}
    TX \circ X 
    = T \left(\Leg^{-1}\right) \circ T \gamma   \circ X
    = T \left(\Leg^{-1}\right) \circ X_{H,\beta} \circ \gamma 
    = \xi_{L,\alpha} \circ \Leg^{-1} \circ \gamma
    = \xi_{L,\alpha} \circ X,
\end{equation}
so $X$ is a weak solution of the Lagrangian Hamilton-Jacobi problem.

Obviously, the Lagrangian weak solution is a solution if and only if the associated Hamiltonian solution is a closed 1-form.
\end{proof}

This result could be extended for regular but not hyperregular Lagrangians (i.e., $\Leg$ is a local diffeomorphism), where it only holds in the open sets where $\Leg$ is a diffeomorphism.

\subsection{Complete solutions}

Complete solutions for the Hamilton-Jacobi problem are defined analogously to the ones in $T^*Q$ (see Definition \ref{def_complete_solutions}).

\begin{definition}
Let $U\subseteq \RR^n$ be an open set. A map $\Phi:Q\times U\to TQ$ is called \emph{complete solution of the Hamilton-Jacobi problem} for $(L,\alpha)$ if
\begin{enumerate}
\item $\Phi$ is a local diffeomorphism,
\item for any $\lambda=(\lambda_1,\ldots,\lambda_n)\in U$, the map
\begin{equation}
\begin{aligned}
  \Phi_\lambda: Q&\to TQ\\
  q&\mapsto \Phi_\lambda(q)=\Phi(q,\lambda_1,\ldots,\lambda_n)
\end{aligned}
\end{equation}
is a solution of the Hamilton-Jacobi problem for $(L,\alpha)$.

\end{enumerate}
\end{definition}

\begin{example}[Fluid resistance]\label{Example_fluid_1D_Lagrangian}
Consider the 1-dimensional Rayleigh system $(L,\mathcal{R})$ \cite{de_leon_symmetries_2021,lopez-gordon_geometry_2021}, with
\begin{equation}
  L=\frac{m}{2}\dot{q}^2, \qquad
  \mathcal{R}=\frac{k}{3}\dot q^3.
\end{equation}
Then
\begin{equation}
  X=\frac{\lambda}{m} e^{-kq/m} \frac{\partial  } {\partial q}
\end{equation}
is a complete solution of the Hamilton-Jacobi problem. Clearly, $\Im X$ is a Lagrangian submanifold.
As a matter of fact,
\begin{equation}
  E_L \circ X = L \circ X =  \frac{\lambda^2}{2m} e^{-2kq/m},
\end{equation}
and
\begin{equation}
  X^*\bar{R} 
  = \frac{k\lambda^2}{m^2} e^{-2kq/m} \dd q
  = -\dd (E_L \circ X).
\end{equation}
The forced Euler-Lagrange vector field is given by
\begin{equation}
  \xi_{L,\bar{R}} = \dot{q} \frac{\partial  } {\partial q} -\frac{k}{m} \dot{q}^2 \frac{\partial  } {\partial \dot{q}},
\end{equation}
whose solutions are
\begin{equation}
  q(t) = \frac{m}{k} \log \left( \frac{\dot{q}_0 kt}{m} +1  \right)+ q_0,
\end{equation}
with initial conditions $q(0)=q_0$ and $\dot{q}(0)=\dot{q}_0$. Similarly, the integral curves of $X$ are given by
\begin{equation}
  q(t) = \frac{m}{k} \log \left( \frac{k\lambda t}{m^2} e^{-kq_0/m} + 1 \right) + q_0.
\end{equation}
Indeed, the integral curves of $\xi_{L,\bar{R}}$ are recovered by taking $\lambda=m\dot{q}_0 e^{kq_0/m}$.
\end{example}

\begin{example} \label{Example_H-J_Lagrangian_fluid_n_dimensions}
Consider the generalization of the previous example to $n$ dimensions, namely
\begin{equation}
  L = \frac{1}{2} \sum_{i=1}^n m_i \left( \dot{q}^i  \right)^2,\qquad
  \mathcal{R} = \frac{1}{3} \sum_{i=1}^n k_i  \left( \dot{q}^i  \right)^3.
\end{equation}
Consider the functions
\begin{equation}
  f_a = m_a e^{k_aq^a/m_a} \dot{q^a},\qquad a=1,\ldots,n
\end{equation}
Locally,
\begin{equation}
  \omega_L = \sum_{i=1}^n m_i \dd q^i\wedge \dd \dot q^i,
\end{equation}
and
\begin{equation}
  X_{f_a} = e^{k_aq^a/m_a} \frac{\partial  } {\partial q^a} 
  - \frac{k_a}{m_a} \dot{q}^a e^{k_aq^a/m_a} \frac{\partial  } {\partial \dot q^a},
\end{equation}
from where it is easy to see that
\begin{equation}
  \left\{f_a, E_L  \right\} = [f_a, \mathcal{R}] = k_a \left( \dot{q}^a  \right)^2 e^{k_aq^a/m_a},
\end{equation}
so, by Eq.~\eqref{double_bracket_Rayleigh}, $f_a$ are constants of the motion.
In addition,
\begin{equation}
  \left\{f_a, f_b  \right\} = 0,
\end{equation}
for every $a,b=1,\ldots,n$, so they are in involution. Then
\begin{equation}
  X=\sum_{i=1}^n \frac{\lambda_i}{m_i} e^{-k_iq^i/m_i} \frac{\partial  } {\partial q^i}
\end{equation}
is a complete solution of the Hamilton-Jacobi problem. See Example \ref{example_Hamiltonian_drag} for the Hamiltonian counterpart of this example.

\end{example}

\begin{example}[Free particle with a homogeneous linear Rayleigh force]
Consider a linear Rayleigh system $(L,\mathcal{R})$, with
\begin{equation}
  L = \frac{1}{2} g_{ij} \dot q^i \dot q^j,   
\end{equation} 
where $g$ does not depend on $q$ (for instance, $g_{ij}=m_i \delta_{ij}$, with $m_i$ the mass of the $i$-th particle), and 
\begin{equation}
  \mathcal{R} = \frac{1}{2} R_{ij} \dot q^i \dot q^j,
\end{equation}
where $R_{ij}$ does not depend on $q$. Then the Hamilton-Jacobi equation can be locally written as
\begin{equation}
  X^k \left(g_{jk} \frac{\partial X^j} {\partial q^i} + R_{ik}  \right) = 0,
\end{equation}
so a complete solution is given by
\begin{equation}
  X_\lambda = \left(\lambda^i - R^i_{j} q^j  \right) \frac{\partial  } {\partial q^i},
\end{equation}
where $R^i_{j} = g^{ik} R_{jk}$. See Example \ref{example_Hamiltonian_homogeneous_Rayleigh} for the Hamiltonian counterpart of this example.

\end{example}
 

\section{Reduction and reconstruction of the Hamilton-Jacobi problem} \label{section_reduction}

Let $G$ be a connected Lie group acting freely and properly on $Q$ by a left action $\Phi$, namely
\begin{equation}
\begin{aligned}
  \Phi: G\times Q &\to Q \\
  (g,q) &\mapsto \Phi (g,q) = g\cdot q. 
\end{aligned}
\end{equation}
As usual, we denote by $\mathfrak{g}$ the Lie algebra of $G$, and denote the dual of $\mathfrak{g}$ by $\mathfrak{g}^*$.
For each $g\in G$, we can define a diffeomorphism
\begin{equation}
\begin{aligned}
  \Phi_g : Q & \to Q\\
  q & \mapsto \Phi_g(q) =  \Phi (g,q) = g\cdot q.
\end{aligned}
\end{equation}
Under these conditions, the quotient space $Q/G$ is a differentiable manifold and $\pi_G: Q \to Q/G$ is a $G$-principal bundle. The action $\Phi$ induces a lifted action $\Phi^{T^*}$ on $T^*Q$ given by
\begin{equation}
  \Phi_g^{T^*}(\alpha_q) = \left(T\Phi_{g^{-1}}  \right)^* (gq) (\alpha_q) \in T_{gq}^*Q,
\end{equation}
for every $\alpha_q\in T_q^*Q$. Since $\Phi$ is a diffeomorphism, its lift to $T^*Q$ leaves $\theta_Q$ invariant \cite{abraham_foundations_2008}, in other words, $\theta_Q$ is $G$-invariant.

The \emph{natural momentum map} $J:T^*Q\to \mathfrak{g}$ is given by
\begin{equation}
  \left\langle J(\alpha_q), \xi  \right\rangle 
  = \left(\contr{\xi_Q^c} \theta_Q \right) (\alpha_q)
\end{equation}
for each $\xi \in \mathfrak{g}$. Here $\xi_Q$ is the infinitesimal generator of the action of $\xi \in \mathfrak{g}$ on $Q$, and $\xi_Q^c$ is the generator of the lifted action on $T^*Q$. The natural momentum map is $G$-equivariant for the lifted action on $T^*Q$. For each $\xi \in \mathfrak{g}$, we have a function $J^\xi:T^*Q\to \RR$ given by
\begin{equation}
  J^\xi (\alpha_q) = \left\langle J(\alpha_q), \xi  \right\rangle,
\end{equation}
that is,
\begin{equation}
  J^\xi = \contr{\xi_Q^c} \theta_Q.
\end{equation}

A vector field $Z$ on $Q$ defines an associated function $\iota Z$ on $T^*Q$. Locally, if
\begin{equation}
  X = X^i \frac{\partial  } {\partial q^i},
\end{equation}
then
\begin{equation}
  \iota X = X^i p_i.
\end{equation}
Given a vector field $X$ on $Q$, its \emph{complete lift} \cite{yano_tangent_1973} is a vector field $X^c$ on $T^*Q$ such that
\begin{equation}
  X^c (\iota Z) = \iota [X, Z]
\end{equation}
for any vector field $Z$ on $Q$. Locally, if $X$ has the form above, then 
\begin{equation}
  X^c = X^i \frac{\partial  } {\partial q^i} - p_j \frac{\partial X^j} {\partial q^i} \frac{\partial  } {\partial p_i}.
\end{equation}

\begin{definition}
  A $G$-invariant forced Hamiltonian system $(H,\beta)$ is a forced Hamiltonian system $(H,\beta)$ such that $H$ and $\beta$ are both $G$-invariant, namely,
  \begin{equation}
    \xi_Q^c (H) = 0,
  \end{equation}
  and
  \begin{equation}
    \liedv{\xi_Q^c} \beta = 0,
  \end{equation}
   for each $\xi \in \mathfrak{g}$.
\end{definition}

If $H$ is known to be $G$-invariant, the subgroup $G_\beta$ such that $(H,\beta)$ is $G_\beta$-invariant can be found through the following lemma.

\begin{lemma}\label{lemma_g_invariance}
 Consider a forced Hamiltonian system $(H,\beta)$. Suppose that $H$ is $G$-invariant. Let $\xi\in \mathfrak{g}$. Then, the following statements hold:
\begin{enumerate}
\item $J^\xi$ is a constant of the motion if and only if
\begin{equation}
  \beta(\xi_Q^c) = 0.
\end{equation}
\item If the previous equation holds, then $\xi$ leaves $\beta$ invariant if and only if
\begin{equation}
  \contr{\xi_Q^c} \dd \beta = 0.
\end{equation}
\end{enumerate}
Moreover, the vector subspace
\begin{equation}
  \mathfrak{g}_\beta = \left\{\xi \in \mathfrak{g} \mid \beta(\xi_Q^c) = 0, \contr{\xi_Q^c} \dd \beta = 0  \right\} \label{subalgebra}
\end{equation}
is a Lie subalgebra of $\mathfrak{g}$.
\end{lemma}

\begin{proof}
Let us prove the first statement. We have that
\begin{equation}
  \dd J^\xi = \dd \left(\contr{\xi_Q^c} \theta_Q  \right)
  = \liedv{\xi_Q^c} \theta_Q - \contr{\xi_Q^c} \dd \theta_Q
  = - \contr{\xi_Q^c} \dd \theta_Q =  \contr{\xi_Q^c} \omega_Q,
\end{equation}
where we have used that $\theta_Q$ is $\mathfrak{g}$-invariant. Contracting with the dynamical vector field $X_{H,\beta}$ (given by Eq.~\eqref{Hamiltonian_dynamics_eq}) yields
\begin{equation}
\begin{aligned}
  X_{H,\beta} (J^\xi) 
  &= \contr{X_{H,\beta}} \dd J^\xi 
  = \contr{X_{H,\beta}} \contr{\xi_Q^c} \omega_Q
  =  - \contr{\xi_Q^c} \contr{X_{H,\beta}} \omega_Q
  = - \contr{\xi_Q^c} (\dd H + \beta)\\
 & = -\xi_Q^c(H) - \beta(\xi_Q^c)
  = -\beta(\xi_Q^c),
\end{aligned}
\end{equation}
since $H$ is $\mathfrak{g}$-invariant, so, by Eq.~\eqref{conserved_quantity_Hamiltonian}, $J^\xi$ is a constant of the motion if and only if $\beta(\xi_Q^c)$ vanishes. The proofs of the other statements are completely analogous to the ones of Lagrangian counterpart of this lemma 
(see references \cite{de_leon_symmetries_2021,lopez-gordon_geometry_2021}; see also reference \cite{bloch_euler-poincare_1996}).
\end{proof}

In the following paragraphs, we shall briefly recall some results we will make use of (see references \cite{de_leon_hamiltonjacobi_2017,marsden_hamiltonian_2007} and references therein for more details).

A \emph{$G$-invariant Lagrangian submanifold} $\mathfrak{L}\subset T^*Q$ is a Lagrangian submanifold in $T^*Q$ such that $\Phi_g^{T^*}(\mathfrak{L}) = \mathfrak{L}$ for all $g\in G$. Given the momentum map $J$ defined above and a Lagrangian submanifold $\mathfrak{L}\subset T^*Q$, it can be shown that $J$ is constant along $\mathfrak{L}$ if and only if $\mathfrak{L}$ is $G$-invariant.

The quotient space $T^*Q/G$ has a Poisson structure induced by the canonical symplectic structure on $T^*Q$, such that $\pi : T^*Q \to T^*Q/G$ is a Poisson morphism.

Denote by $G_\beta\subset G$ the Lie subgroup whose Lie algebra is $\mathfrak{g}_\beta$, defined by \eqref{subalgebra}. Let $\mu\in \mathfrak{g}^*$. Let us assume that $(G_\beta)_\mu = G$.
Then $J^{-1}(\mu)/G$ is a symplectic leaf of $T^*Q/G$. Moreover, $J^{-1}(\mu)$ is a coisotropic submanifold. Assume that $\mathfrak{L}\subset J^{-1}(\mu)$ is a Lagrangian submanifold. Then, by the Coisotropic Reduction Theorem \cite{abraham_foundations_2008}, $\pi(\mathfrak{L})$ is a Lagrangian submanifold to $J^{-1}(\mu)/G$.

Furthermore, it can be shown that $J^{-1}(\mu)/G$ is diffeomorphic to the cotangent bundle $T^*(Q/G)$. In addition, if $J^{-1}(\mu)/G$ is endowed with the symplectic structure $\omega_\mu$ given by the Marsden-Weinstein reduction procedure, it is symplectomorphic to $T^*(Q/G)$ endowed with a modified symplectic structure $\tilde{\omega}_{Q/G}$. This modified symplectic form is given by the canonical symplectic form plus a magnetic term, namely
\begin{equation}
  \tilde{\omega}_{Q/G} = \omega_{Q/G} + B_\mu,
\end{equation}
where $\omega_{Q/G}$ is the canonical symplectic form on $T^*(Q/G)$.

Combining the previous paragraphs, $\pi(\mathfrak{L})$ can be seen as a Lagrangian submanifold of a cotangent bundle with a modified symplectic structure. Let $\tilde{\mathfrak{g}}^*$ denote the adjoint bundle to $\pi_G:Q\to Q/G$ via the coadjoint representation, 
\begin{equation}
  \tilde{\mathfrak{g}}^* = Q\times_{G} \mathfrak{g}^*.
\end{equation}
A connection $A$ on $\pi_G:Q\to Q/G$ induces a splitting
\begin{equation}
  T^*Q/G \equiv T^*(Q/G)  \times_{Q/G} \tilde{\mathfrak{g}}^*.
\end{equation}
This identification is given by
\begin{equation}
\begin{aligned}
  \Psi: T^*Q/G &\to T^*(Q/G)  \times_{Q/G} \tilde{\mathfrak{g}}^*  \\
  [\alpha_q] & \mapsto \left[ \left( \alpha_q \circ \hor_q, J(\alpha_q)  \right)  \right],
\end{aligned}
\end{equation}
where $\hor_q: T_{\pi_G(q)}(Q/G)\to \mathcal{H}_q$ denotes the horizontal lift of the connection $A$.
Here $\mathcal{H}_q$ denotes the horizontal space of $A$ at $q\in Q$, namely
\begin{equation}
  \mathcal{H}_q = \left\{v_q \in T_q Q \mid A(v_q) = 0  \right\}.
\end{equation}
If $\alpha_q\in J^{-1}(\mu)$, then $J(\alpha_q)=\mu$ and $\Psi([\alpha_q])=(\alpha_q \circ \hor_q, \mu)$, so
\begin{equation}
  \Psi \left(J^{-1}(\mu)/G  \right) = T^*(Q/G) \times_{Q/G} \left(Q\times \left\{\mu  \right\}/G  \right)
  \equiv T^*(Q/G).
\end{equation}

Consider a $G$-invariant forced Hamiltonian system $(H,\beta)$ on $T^*Q$. Then $H=H_G\circ \pi$, where $H_G:T^*Q/G\to \RR$ is the reduced Hamiltonian on $T^*Q/G$. Similarly, $\beta = \pi^* \beta_G$, where $\beta_G$ is the reduced external force on $T^*Q/G$. Moreover, we can define the reduced Hamiltonian $\tilde{H}_\mu$ on $T^*(Q/G)$ by
\begin{equation}
  \tilde{H}_\mu (\tilde{\alpha}_{\tilde{q}}) = \tilde{H} (\tilde{\alpha}_{\tilde{q}}, [q,\mu]),
\end{equation}
for each $(\tilde{\alpha}_{\tilde{q}})\in T_{\tilde q}^*(Q/G)$, where $\tilde{q} = [q] \in Q/G$ and $\tilde{H}=H_G\circ \Psi^{-1}$. Similarly, let $\tilde{\beta}= (\psi^{-1})^* \beta_G$, and
\begin{equation}
  \tilde{\beta}_\mu (\tilde{\alpha}_{\tilde{q}}) = \tilde \beta (\tilde{\alpha}_{\tilde{q}}, [q,\mu]),
\end{equation}
for each $(\tilde{\alpha}_{\tilde{q}})\in T_{\tilde q}^*(Q/G)$.

Let $\gamma$ be a $G$-invariant solution of the Hamilton-Jacobi problem for $(H,\beta)$. Then, the following diagram commutes:
\begin{center}
\begin{tikzcd}
                                                                                                     & \tilde{\mathfrak{g}}^*                                                              &                                                                                           &                                                                                                                                                         \\
Q \arrow[r, "\gamma"] \arrow[d, "\pi_G"] \arrow[rrdd, "\tilde \gamma", bend right=67, shift right=3] & T^*Q \arrow[l, "\hspace{-.25cm} \pi_Q"', bend left] \arrow[d, "\pi"] \arrow[rr, "H"] \arrow[u, "J"] &                                                                                           & \mathbb{R}                                                                                                                                              \\
\frac Q G \arrow[rrd, "\tilde \gamma_\mu"]                                                           & \frac {T^*Q}{G} \arrow[l, "p"'] \arrow[rr, "\psi", shift left] \arrow[rru, "H_G"]   &                                                                                           & \hspace{-.1cm} T^*\left(\frac Q G\right)  \hspace{-.1cm}   \times_{ \hspace{-.1cm}\frac{Q}{G}} \hspace{-.07cm}\tilde{\mathfrak{g}}^* \hspace{-.1cm}   \arrow[ll, "\psi^{-1}", shift left] \arrow[u, "\tilde H"] \arrow[ld, dashed] \\
                                                                                                     &                                                                                     & T^*\left(\frac {Q} {G}\right) \arrow[ruu, "\tilde{H}_\mu"', bend right=90, shift right=6] &                                                                                                                                                        
\end{tikzcd}
\end{center}

\begin{proposition}[Reduction]
Let $\gamma$ be a $G$-invariant solution of the Hamilton-Jacobi problem for $(H,\beta)$. Let $\mathfrak{L} = \Im \gamma$, and $\tilde{\mathfrak{L}}=\Psi\circ \pi(L)$. Then $\gamma$ induces a mapping $\tilde \gamma_\mu$ such that $\Im \tilde \gamma_\mu = \tilde {\mathfrak{L}}$ and $\tilde \gamma_\mu$ is a solution the Hamilton-Jacobi problem for $(\tilde{H}_\mu, \tilde{\beta}_\mu)$.
\end{proposition}

\begin{proof}
Since $\dd H +\beta$ vanishes along $\mathfrak{L}$, clearly $\dd H_G + \beta_G$ vanishes along $\pi(\mathfrak{L})$. If $\tilde{\alpha}_{\tilde q}\in \tilde{\mathfrak{L}}$, then $\psi^{-1}(\tilde{\alpha}_{\tilde q}) \in \pi (\mathfrak{L})$, 
\begin{equation}
  \tilde{H}_\mu (\tilde{\alpha}_{\tilde q}) = H_G \circ \Psi^{-1} (\tilde{\alpha}_{\tilde q}, \mu),
\end{equation}
and
\begin{equation}
  \tilde{\beta}_\mu (\tilde{\alpha}_{\tilde{q}}) = (\Psi^{-1})^* \beta_G (\tilde{\alpha}_{\tilde{q}}, \mu),
\end{equation}
so
\begin{equation}
  \left(\tilde{\beta}_\mu + \dd \tilde{H}_\mu   \right) (\tilde{\alpha}_{\tilde{q}})
  = \left(\dd H_G + \beta_G  \right) \left(\Psi^{-1}(\tilde{\alpha}_{\tilde{q}}, \mu)\right) = 0.
\end{equation}
Since $\gamma$ is $G$-invariant, it induces a mapping $\tilde{\gamma}: Q \to T^*(Q/G)$ which is also $G$-invariant. This mapping in turn induces a reduced solution $\tilde \gamma_\mu:Q/G\to T^*(Q/G)$ such that $\tilde{\gamma}=\pi_G^*\tilde \gamma_\mu$ and $\Im \tilde \gamma_\mu = \tilde{\mathfrak{L}}$. 
\end{proof}

\begin{proposition}[Reconstruction]
Let $\tilde{\mathfrak{L}}$ be a Lagrangian submanifold of $(T^*(Q/G), \omega_{Q/G}+B_\mu)$ for some $\mu\in \mathfrak{g}^*$ which is a fixed point of the coadjoint action. Assume that $\tilde{\mathfrak{L}} = \Im \tilde \gamma_\mu$, where $\tilde \gamma_\mu$ is a solution of the Hamilton-Jacobi problem for $(\tilde H_\mu, \tilde \beta_\mu)$. Let 
\begin{equation}
  \hat{\mathfrak{L}} = \left\{ (\tilde{\alpha}_{\tilde q}, [\mu]_{\tilde q}) \in T^*(Q/G) \times_{Q/G} \tilde{\mathfrak{g}}^*\mid \tilde{\alpha}_{\tilde q} \in \tilde{\mathfrak{L}}  \right\},
\end{equation}
and take
\begin{equation}
  \mathfrak{L} = \pi^{-1} \circ \Psi^{-1} (\hat{\mathfrak{L}}).
\end{equation}
Then
\begin{enumerate}
\item $\mathfrak{L}$ is a $G$-invariant Lagrangian submanifold of $T^*Q$,
\item $\mathfrak{L}= \Im \gamma$, where $\gamma$ is a solution of the Hamilton-Jacobi problem for $(H, \beta)$.
\end{enumerate}
\end{proposition}

\begin{proof}
The mapping $\tilde{\gamma}_\mu: Q/G\to T^*(Q/G)$ induces a $G$-invariant mapping
\begin{equation}
  \tilde{\gamma}=\pi_G^*\tilde \gamma_\mu:Q \to T^*(Q/G),
\end{equation}
which in turn induces a $G$-invariant mapping $\gamma:Q\to T^*Q$. Let $\tilde \alpha_{\tilde q} \in \tilde{\mathfrak{L}}$, so $(\tilde{\alpha}_{\tilde q}, [\mu]_{\tilde q}) \in T^*(Q/G) \times_{Q/G} \tilde{\mathfrak{g}}^*$, and 
\begin{equation}
  \pi^{-1} \circ \Psi^{-1} (\tilde{\alpha}_{\tilde q}, [\mu]_{\tilde q}) \in \mathfrak{L},
\end{equation}
and then
\begin{equation}
  \left(\tilde{\beta}_\mu + \dd \tilde{H}_\mu   \right) (\tilde{\alpha}_{\tilde{q}})
  = \left(\dd H_G + \beta_G  \right) \left(\Psi^{-1}(\tilde{\alpha}_{\tilde{q}}, \mu)\right) = 0.
\end{equation}

\end{proof}

\begin{example}[Calogero-Moser system with a linear Rayleigh force]
Consider the linear Rayleigh system $(H, \hat{R})$ on $T^*\RR^2$, with
\begin{equation}
  H = \frac{1}{2} \left(p_1^2 + p_2^2 + \frac{1}{(q^1-q^2)^2}  \right),
\end{equation}
and
\begin{equation}
  \hat{R} = \left(\frac{\partial } {\partial q^1 } + \frac{\partial  } {\partial q^2}  \right) \otimes \left(\dd q^1 - \dd q^2  \right),
\end{equation}
so
\begin{equation}
  \tilde{R} = (p_1 + p_2) (\dd q^1- \dd q^2).
\end{equation}
Consider the action of $\RR$ on $\RR^2$ given by
\begin{equation}
\begin{aligned}
  \Phi : \RR \times \RR^2 &\to \RR^2\\
  \left(r,(q^1, q^2)  \right) & \mapsto (r+q^1, r+q^2).  
\end{aligned}
\end{equation}
Clearly, $(H, \tilde{R})$ is invariant under the corresponding lifted action on $T^*Q$. The momentum map is
\begin{equation}
\begin{aligned}
  J: T^*\RR^2 &\to \RR\\
  (q^1, q^2, p_1, p_2) & \mapsto p_1 + p_2.
\end{aligned}
\end{equation}
Then
\begin{equation}
  J^{-1}(\mu) = \left\{(q^1, q^2, p_1, p_2) \in T^*\RR^2 \mid p_2 = \mu - p_1  \right\}.
\end{equation}
We can identify $J^{-1}(\mu)/\RR$ with $\RR^2$, with coordinates $(q,p)$ and the natural projection
\begin{equation}
\begin{aligned}
    \pi : J^{-1}(\mu) &\to J^{-1}(\mu)/\RR\\
    (q^1, q^2, p, \mu-p) &\mapsto (q=q^1-q^2, p).
\end{aligned}
\end{equation}
We can introduce a reduced Hamiltonian
\begin{equation}
  h = \frac{1}{2} \left((\mu-p)^2 + p^2 + \frac{1}{q^2}\right),  
\end{equation}
and a reduced external force
\begin{equation}
  r = \mu \dd q.
\end{equation}
Let $\tilde{\gamma}$ be a closed 1-form on $\RR$. Then
\begin{equation}
  \tilde{\gamma}^*(\dd h + r) = \left(\mu - \frac{1}{q^3} -\mu \frac{\partial \tilde{\gamma}} {\partial q} \right) \dd q,
\end{equation}
so $\tilde{\gamma}$ is a solution of the Hamilton-Jacobi problem for $(h, r)$ if and only if
\begin{equation}
  \mu - \frac{1}{q^3} -\mu \frac{\partial \tilde{\gamma}} {\partial q} = 0,
\end{equation}
and hence
\begin{equation}
  \tilde{\gamma}_\lambda = \left(q + \frac{1}{2\mu q^2} + \lambda\right) \dd q
\end{equation}
is a complete solution depending on the parameter $\lambda \in \RR$. The associated generating function is
\begin{equation}
  \tilde S_\lambda (q) = \frac{1}{2} q^2 - \frac{1}{2\mu q} + \lambda q ,
\end{equation}
where, without loss of generality, we have taken the integration constant as zero. We can reconstruct a complete solution $\gamma_\lambda$ of the Hamilton-Jacobi for $(H, \tilde R)$, given by
\begin{equation}
  \gamma_\lambda =  \left(q^1 - q^2 + \frac{1}{2\mu (q^1-q^2)^2} + \lambda\right) \dd q^1
  +  \left(\mu - q^1 + q^2 - \frac{1}{2\mu (q^1-q^2)^2} - \lambda\right) \dd q^2,
\end{equation}
and the associated generating function is
\begin{equation}
  S_\lambda (q^1, q^2) = \tilde S_\lambda(q^1-q^2) + \mu q^2.
\end{equation}
\end{example}

\section{\v{C}aplygin systems} \label{section_Chaplygin}

A nonholonomic mechanical system is given by a Lagrangian function $L=L\left(q^{A}, \dot{q}^{A}\right)$ subject to a family of constraint functions
\begin{equation}
  \Phi^{i}\left(q^{A}, \dot{q}^{A}\right)=0,1 \leq i \leq m \leq n=\operatorname{dim} Q.  
\end{equation}
For the sake of simplicity, we shall assume that the constraints $\Phi^{i}$ are linear in the velocities, i.e., $\Phi^{i}\left(q^{A}, \dot{q}^{A}\right)=\Phi_{A}^{i}(q) \dot{q}^{A}$. Then the nonholonomic equations of motion are
\begin{equation}
\begin{aligned}
  &\frac{\dd}{\dd t}\left(\frac{\partial L}{\partial \dot{q}^{A}}\right)-\frac{\partial L}{\partial q^{A}} =\lambda_{i} \Phi_{A}^{i}(q), \quad &1 \leq A \leq n, \\
  &\Phi^{i}\left(q^{A}, \dot{q}^{A}\right) =0, \quad &1 \leq i \leq m,
\end{aligned}
\end{equation}
where $\lambda_{i}=\lambda_{i}\left(q^{A}, \dot{q}^{A}\right), 1 \leq i \leq m$, are Lagrange multipliers to be determined.

Geometrically, the constraints are given by a vector subbundle $M$ of $T Q$ locally defined by $\Phi^{i}=0$. The dynamical equations can then be rewritten intrinsically as
\begin{equation}
\begin{aligned}
      &\contr{X} \omega_{L}-\dd E_{L} \in S^{*}\left((T M)^{0}\right), \\
      &X  \in T M.
\end{aligned}
\end{equation}

Under certain compatibility conditions, the vector field $X$ is unique and it is denoted by $X_{\mathrm{nh}}$.

A \emph{\v{C}aplygin system} (also spelled as Chaplygin) is a nonholonomic mechanical system such that:
\begin{enumerate}
\item the configuration manifold $Q$ is a fibred manifold, say $\rho: Q \longrightarrow N$, over a manifold $N$;
\item the constraints are provided by the horizontal distribution of an Ehresmann connection $\Gamma$ in $\rho$;
\item the Lagrangian $L: T Q \longrightarrow \mathbb{R}$ is $\Gamma$-invariant.
\end{enumerate}

A particular case is when $\rho: Q \longrightarrow N=Q / G$ is a principal $G$-bundle and $\Gamma$ a principal connection. As a matter of fact, some references \cite{cantrijn_geometry_2002,bloch_nonholonomic_1996,cantrijn_reduction_1999,koiller_reduction_1992} restrict their definition of \emph{\v{C}aplygin system} to this particular case. Our more general definition is also considered in references \cite{de_leon_geometry_1996,iglesias-ponte_towards_2008}.

Let us recall that the connection $\Gamma$ induces a Whitney decomposition $T Q=\mathcal{H} \oplus V \rho$ where $\mathcal{H}$ is the horizontal distribution, and $V \rho=\mathrm{ker} T \rho$ is the vertical distribution. Take fibred coordinates $\left(q^{A}\right)=\left(q^{a}, q^{i}\right)$ such that $\rho\left(q^{a}, q^{i}\right)=\left(q^{a}\right)$.

With a slight abuse of notation, let $\hor: T Q \to \mathcal{H}$ denote the horizontal projector hereinafter, and let $\Gamma_{a}^{i}=\Gamma_{a}^{i}\left(q^{A}\right)$ the Christoffel components of the connection $\Gamma$. Let us recall that $\Gamma$ may be considered as a $(1,1)$-type tensor field on $Q$ with $\Gamma^2 = \mathrm{id}$, so $\hor = \hor^2 = (1/2)(\mathrm{id}+\Gamma)$. The curvature of $\Gamma$ is the $(1,2)$-tensor field $\mathfrak{R}=\frac{1}{2}[\hor, \hor]$ where $[\hor, \hor]$ is the Nijenhuis tensor of $\hor$ \cite{leon_methods_1989}, that is
\begin{equation}
  \mathfrak{R}(X, Y)=[\hor X, \hor Y]-\hor[\hor X, Y]-\hor[X, \hor Y]+\hor^{2}[X, Y],
\end{equation}
for each pair of vector fields $X$ and $Y$ on $N$. Locally,
\begin{equation}
  \mathfrak{R}\left(\frac{\partial}{\partial q^{a}}, \frac{\partial}{\partial q^{b}}\right)=\mathfrak{R}_{a b}^{i} \frac{\partial}{\partial q^{i}},
\end{equation}
where
\begin{equation}
  \mathfrak{R}_{a b}^{i}=\frac{\partial \Gamma_{a}^{i}}{\partial q^{b}}-\frac{\partial \Gamma_{b}^{i}}{\partial q^{a}}+\Gamma_{a}^{j} \frac{\partial \Gamma_{b}^{i}}{\partial q^{j}}-\Gamma_{b}^{j} \frac{\partial \Gamma_{a}^{i}}{\partial q^{j}}.
\end{equation}
The constraints are given by
\begin{equation}
  \Phi^{i}=\dot{q}^{i}+\Gamma_{a}^{i} \dot{q}^{a}=0,
\end{equation}
 that is, the solutions are horizontal curves with respect to $\Gamma$.

Since the Lagrangian is $\Gamma$-invariant,
\begin{equation}
  L \left((Y^{\mathcal{H}})_{q_1}  \right) =  L \left((Y^{\mathcal{H}})_{q_2}  \right),
\end{equation}
for all $Y\in T_yN,\ y=\rho(q_1)=\rho(q_2)$, where $Y^\mathcal{H}$ denotes the horizontal lift of $Y$ to $Q$. We can then introduce a function $\ell$ on $TN$ such that
 \begin{equation}
   \ell(Y_y) = L \left((Y^{\mathcal{H}})_q  \right),
 \end{equation}
  so locally we have
 \begin{equation}
   \ell\left(q^{a}, \dot{q}^{a}\right)=L\left(q^{a}, q^{i}, \dot{q}^{a},-\Gamma_{a}^{i} \dot{q}^{a}\right).
 \end{equation}
Now let $\rho(q)=y$ and $x \in \mathcal{H}$ with $\tau_{Q}(x)=q$; 
let $u\in T_y N,\ U\in T_u(TN)$ and $X\in  T_x(TQ)$ such that $X$ projects onto
\begin{equation}
  \mathfrak{R}\left(\left(u^{\mathcal H}\right)_{q},\left(T \tau_{M}(U)\right)_{q}^{\mathcal H}\right) \in T_{q} Q.
\end{equation}
We can then introduce a 1-form $\upalpha$ on $TN$ such that
\begin{equation}
 \left( \upalpha  \right)_u(U)= -\left(\theta_L  \right)_u (X),
\end{equation}
where $\theta_L$ is the Poincaré-Cartan 1-form associated with $L$, given by Eq.~\eqref{Poincare_Cartan_1}.
In other words, $\upalpha$ is locally given by
\begin{equation}
  \upalpha=\left(\frac{\partial L}{\partial \dot{q}^{i}} \dot{q}^{b} \mathfrak{R}_{a b}^{i}\right) \dd q^{a}.
\end{equation}
It can be shown that $\ell$ is a regular Lagrangian, and that the \v{C}aplygin system is equivalent to the forced Lagrangian system $(\ell,\upalpha)$.

Assume that $\ell$ is hyperregular. Then, the \v{C}aplygin system has an associated forced Hamiltonian system $(h, \upbeta)$, with
\begin{equation}
  h = \ell \circ \Leg_{\ell},
\end{equation}
and
\begin{equation}
  \upbeta = \left( \Leg_{\ell}^{-1} \right)^* \upalpha.
\end{equation}
In particular, if $L$ is natural we have
\begin{equation}
  \upbeta =  g^{jb} p_i p_j \mathfrak{R}_{a b}^{i} \dd q^{a}.
\end{equation}
The Hamilton-Jacobi equation for $(h, \upbeta)$ is thus locally
\begin{equation}
  \frac{\partial V} {\partial q^a} + g^{jk} \gamma_k \frac{\partial  \gamma_j} {\partial q^a}
  + g^{jb} \gamma_k \gamma_j \mathfrak{R}^k_{ab} = 0.
\end{equation}
In particular, if $L$ is purely kinetical, 
\begin{equation}
  g^{jk}  \frac{\partial  \gamma_j} {\partial q^a}
  + g^{jb}  \gamma_j \mathfrak{R}^k_{ab} = 0.
\end{equation}

Let $D$ denote a distribution on $Q$ whose annihilator is 
\begin{equation}
   D^0 = \operatorname{span} \left\{ \mu^i = \Phi^i_A (q)\dd q^A \right\}.
 \end{equation} 
 Then, we can form the algebraic ideal $\mathcal{I}(D^0)$ in the algebra $\Lambda^*(Q)$, namely, if a $k$-form $\nu \in \mathcal{I}(D^0)$, then
 \begin{equation}
    \nu = \upbeta_i \wedge \mu^i,
 \end{equation}
 where $\upbeta_i\in \Lambda^{k-1}(q)$ and $1\leq i \leq m$.

 \begin{theorem}[Iglesias-Ponte, de León, Martín de Diego {\cite[Theorem 4.3]{iglesias-ponte_towards_2008}}]
 Let $\mathcal{H}$ denote the horizontal distribution defined by the connection $\Gamma$ in $\rho:Q\to N$.
 Let $X$ be vector field on $Q$ such that $X(Q) \subset \mathcal{H}$ and $\dd (\Leg \circ X) \in \mathcal{I}(\mathcal{H}^0)$. Then the following conditions are equivalent:
 \begin{enumerate}
 \item for every curve $\sigma: \mathbb{R} \to Q$ such that
 \begin{equation}
   \dot{\sigma}(t)=T \tau_{Q}\circ X_{\mathrm{nh}}\circ X\circ \sigma(t)
 \end{equation}
 for all $t$, then $X \circ \sigma$ is an integral curve of $X_{\mathrm{nh}}$;
 \item $\dd\left(E_{L} \circ X\right) \in \mathcal{H}^{0}$.
 \end{enumerate}
 A vector field $X$ satisfying these conditions is called a solution of the nonholonomic Hamilton-Jacobi problem for $(L, \Gamma)$.
 \end{theorem}

\begin{theorem}[Iglesias-Ponte, de León, Martín de Diego {\cite[Theorem 4.5]{iglesias-ponte_towards_2008}}]
Assume that a vector field $X$ on $Q$ is a solution for the nonholonomic Hamilton-Jacobi problem for $(L, \Gamma)$. If $X$ is $\rho$-projectable to a vector field $Y$ on $N$ and $\gamma = \left(\Leg_{\ell}  \right)_{*}Y$ is closed, then $Y$ is a solution of the Lagrangian Hamilton-Jacobi problem for $(\ell, \upalpha)$, and $\gamma$ is a solution of the Hamilton-Jacobi problem for $(h,\upbeta)$.

Conversely, let $\gamma$ be a solution of the Hamilton-Jacobi problem for $(h, \upbeta)$. Then $Y=\Leg^{-1}\circ \gamma$ is a solution of the Lagrangian Hamilton-Jacobi problem for $(\ell, \upalpha)$. Then, if 
\begin{equation}
  \dd\left(\Leg_L \circ Y^{\mathcal{H}}\right) \in \mathcal{J}\left(\mathcal{H}^{0}\right),
\end{equation}
then the horizontal lift $Y^{\mathcal{H}}$ is a solution for the nonholonomic Hamilton-Jacobi problem for $(L, \Gamma)$.

\end{theorem}

\begin{example}[Mobile robot with fixed orientation]

Consider the motion of a robot whose body maintains a fixed orientation with respect to the environment. The robot has three wheels with radius $R$, which turn simultaneously about independent axes, and perform a rolling without sliding over a horizontal floor (see references \cite{cantrijn_geometry_2002,iglesias-ponte_towards_2008,kelly_geometric_1995} for more details).

Let $(x, y)$ denotes the position of the centre of mass, let $\theta$ and $\psi$ denote the steering and rotation angles of the wheels, respectively. Hence, the configuration manifold is $Q=S^{1} \times S^{1} \times \mathbb{R}^{2}$, and the Lagrangian of the system is
\begin{equation}
  L=\frac{1}{2} m \dot{x}^{2}+\frac{1}{2} m \dot{y}^{2}+\frac{1}{2} J \dot{\theta}^{2}+\frac{3}{2} J_{\omega} \dot{\psi}^{2}.
\end{equation}
Here $m$ is the mass, $J$ is the moment of inertia and $J_{\omega}$ is the axial moment of inertia of the robot.

The constraints are induced by the conditions that the wheels roll without sliding, in the direction in which they point, and that the instantaneous contact point of the wheels with the floor have no velocity component orthogonal to that direction, so we have
\begin{equation}
\begin{aligned}
  &\dot{x} \sin \theta-\dot{y} \cos \theta=0, \\
  & \dot{x} \cos \theta+\dot{y} \sin \theta-R \dot{\psi}=0 .
\end{aligned}
\end{equation}
The constraint distribution $\mathcal{D}$ is spanned by
\begin{equation}
    \left\{\frac{\partial}{\partial \theta}, \frac{\partial}{\partial \psi}+R\left(\cos \theta \frac{\partial}{\partial x}+\sin \theta \frac{\partial}{\partial y}\right)\right\}
\end{equation}
The Abelian group $G=\mathbb{R}^{2}$ acts on $Q$ by translations, namely
\begin{equation}
    ((a, b),(\theta, \psi, x, y)) \mapsto(\theta, \psi, a+x, b+y).
\end{equation}
Therefore, we have a principal $G$-bundle $\rho: Q \longrightarrow N=Q / G$ with a principal connection  given by the $\mathfrak{g}$-valued 1-form
\begin{equation}
    \eta =(\mathrm{d} x-R \cos \theta \mathrm{d} \psi) e_{1}+(\mathrm{d} y-R \sin \theta \mathrm{d} \psi) e_{2},
\end{equation}
where $\mathfrak{g}=\RR^2$ is the Lie algebra of $G$, and $\left\{e_{1}, e_{2}\right\}$ is the canonical basis of $\mathbb{R}^{2}$ (identified with $\mathfrak{g}$).

One can show that the reduced forced mechanical system $(\ell, \upalpha)$ on $TN$ is given by
\begin{equation}
  \ell(\theta, \psi, \dot \theta, \dot \psi)=\frac{1}{2} J \dot{\theta}^{2}+\frac{m R^{2}+3 J_{\omega}}{2} \dot{\psi}^{2},
\end{equation}
and $\upalpha$ identically zero. A complete solution of the Hamilton-Jacobi problem for $(\ell, \upalpha)$
is given by
\begin{equation}
\begin{aligned}
  Y_\lambda = \lambda_\theta \frac{\partial  } {\partial \theta} + \lambda_\psi \frac{\partial  } {\partial \psi}.  
\end{aligned}
\end{equation}
Its horizontal lift is 
\begin{equation}
  Y_\lambda^{\mathcal H} = \lambda_\theta \frac{\partial  } {\partial \theta} 
  + \lambda_\psi \left( \frac{\partial  } {\partial \psi} + R \cos \theta \frac{\partial  } {\partial x}
   + R \sin \theta \frac{\partial  } {\partial y}  \right),
\end{equation}
so
\begin{equation}
  \gamma_{\lambda} = \Leg_L \circ  Y_\lambda^{\mathcal H} 
  = \lambda_\theta J \dd \theta 
  + \lambda_\psi \left(    3 J_{\omega} \mathrm{d} \psi + m R \cos \theta \mathrm{d} x+m R \sin \theta \mathrm{d} y \right),
\end{equation}
and thus
\begin{equation}
  \dd \gamma_\lambda = -\lambda_\psi m R \ \mathrm{d} \theta \wedge(\sin \theta \mathrm{d} x-\cos \theta \mathrm{d} y) \in \mathcal{I}\left(\mathcal{H}^{0}\right).
\end{equation}
Hence, $Y^{\mathcal H}_\lambda$ is a complete solution of the Hamilton-Jacobi problem for $(L, \Gamma)$.
\end{example}

\section{Conclusions and outlook}

In this paper we have obtained a Hamilton-Jacobi theory for Hamiltonian and Lagrangian systems with external forces. We have discussed the complete solutions of the Hamilton-Jacobi problem. Our results have been particularized for forces of Rayleigh type. We have presented a dissipative bracket for Rayleigh systems. Furthermore, we have studied the reduction and reconstruction of the Hamilton-Jacobi problem for forced Hamiltonian systems with symmetry. Additionally, we have shown how the Hamilton-Jacobi problem for a \v{C}aplygin system can be reduced to the Hamilton-Jacobi problem for a forced Lagrangain system, in order to obtain solutions of the latter and reconstruct solutions of the former.

In a previous paper we studied the symmetries, conserved quantities and reduction of forced mechanical systems (see reference \cite{de_leon_symmetries_2021}, see also reference \cite{lopez-gordon_geometry_2021}). Making use of results from this paper, one can obtain the constants of the motion in involution of a forced system, and relate them with complete solutions of the Hamilton-Jacobi problem for that system (see Example \ref{Example_fluid_1D_Lagrangian}). Furthermore, Lemma 15 from reference \cite{de_leon_symmetries_2021} has been translated to the Hamiltonian formalism (see Lemma \ref{lemma_g_invariance}), in order to extend the method of reduction of the Hamilton-Jacobi problem \cite{de_leon_hamiltonjacobi_2017} for forced Hamiltonian systems.

In another paper \cite{de_leon_discrete_2021}, we develop a Hamilton-Jacobi theory for forced discrete Hamiltonian systems. Our approach is based on the construction of a discrete flow on $Q\times Q$ (unlike the case without external forces \cite{de_leon_geometry_2018}, where the discrete flow is defined on $Q$). We define a discrete Rayleigh potential. Additionally, we present some simulations and analyse their numerical accuracy.

An additional open problem is the particularization of the results from this paper when the configuration space $Q$ is a Lie group $G$ with Lie algebra $\mathfrak g$.
If $(L,\alpha)$ is a $G$-invariant forced Lagrangian system on $TG$, the forced Euler-Lagrange equations for $(L, \alpha)$ are reduced to the Euler-Poincaré equations with forcing on $\mathfrak{g}$ (see reference \cite{bloch_euler-poincare_1996}).
We also plan to extend our results on forced systems to the Lie algebroid framework, in order to use Atiyah algebroids when the system enjoys symmetries.
Furthermore, we plan to extend the results from this paper for time-dependent forced Lagrangian systems in the framework of cosymplectic geometry (see reference \cite{colombo_symplectic_2021}). Additionally, the applications of the dissipative bracket \eqref{dissipative_bracket} will be studied elsewhere.


\section*{Acknowledgements}
We thank the referee for his/her constructive comments.
The authors acknowledge financial support from the Spanish Ministry of Science and Innovation (MICINN), under grants PID2019-106715GB-C21 and ``Severo Ochoa Programme for Centres of Excellence in R\&D'' (CEX2019-000904-S). Manuel Lainz wishes to thank MICINN and the Institute of Mathematical Sciences (ICMAT) for the FPI-Severo Ochoa predoctoral contract PRE2018-083203. Asier López-Gordón would like to thank MICINN and ICMAT for the predoctoral contract PRE2020-093814. He is also grateful for enlightening discussions on fibre bundles with his friend and colleague Alejandro Pérez-González.

\section*{Data Availability}
Data sharing is not applicable to this article as no new data were created or analysed in this study.

\printbibliography

\end{document}